\documentclass[10pt]{article}

\usepackage[default]{jasa_harvard}    
\usepackage{JASA_manu}

\newcommand {\ctn}{\citeasnoun} 
\newcommand {\ctp}{\cite}       

\RequirePackage[OT1]{fontenc}
\RequirePackage{graphicx,amsthm,amsmath,latexsym,amssymb}
\RequirePackage{float,epsfig,multirow,rotating,times}

\pdfoutput=1
\numberwithin{equation}{section}
\theoremstyle{plain}
\newtheorem{theorem}{Theorem}[section]

\newcommand{\bv}{\boldsymbol{v}}

\begin{document}
\title{\large{\bf{Bayesian Density Estimation via Multiple Sequential Inversions of 2-D Images with Application in Electron Microscopy}}\protect} 
\author{{{{Dalia} {Chakrabarty}$^{1,}$$^{6}$, University of Warwick, {d.chakrabarty@warwick.ac.uk and}}}\\ 
{{{University of Leicester, dc252@le.ac.uk}}}\\
{{{Fabio} {Rigat}{$^{2}$}, Novartis Vaccines and Diagnostics, {fabio.rigat@novartis.com}}}\\
{{{Nare} {Gabrielyan}{$^{3}$}, De Montfort University, {nare.gabrielyan@email.dmu.ac.uk}}}\\
{{{Richard} {Beanland}{$^{4}$}, University of Warwick, {r.beanland@warwick.co.uk}}}\\
{{{Shashi} {Paul}{$^{5}$}, De Montfort University, {spaul@dmu.ac.uk}}}}
\date{}
\maketitle

\mbox{}
\vspace*{1.5in}
\begin{center}
\textbf{Authors' affiliations:}
\end{center}
{\small{$^{1}${Associate Research Fellow in Department of Statistics,
  University of Warwick, Coventry CV4 7AL, U.K., corresponding author}, \\
{$^{6}$}{Lecturer of Statistics, Department of Mathematics, University of Leicester, Leicester LE1 7RH,  U.K.},\\
{$^{2}$}{Research Biostatistics Group Head, Novartis Vaccines
  and Diagnostics and Associate fellow at Department of Statistics,
  University of Warwick},\\
{$^{3}$}{Graduate student at Emerging Technologies Research
  Centre, De Montfort University},\\
{$^{4}$}{Lecturer, Department of Physics, University of Warwick},\\
{$^{5}$}{Head, Emerging Technologies Research Centre, De
  Montfort University}}}

\newpage
\begin{center}
\textbf{Abstract}
\end{center}
We present a new Bayesian methodology to learn the unknown material
density of a given sample by inverting its two-dimensional images that
are taken with a Scanning Electron Microscope. An image results from a
sequence of projections of the convolution of the density function
with the unknown microscopy correction function that we also learn
from the data. We invoke a novel design of experiment, involving
imaging at multiple values of the parameter that controls the
sub-surface depth from which information about the density structure
is carried, to result in the image. Real-life material density
functions are characterised by high density contrasts and typically
are highly discontinuous, implying that they exhibit correlation
structures that do not vary smoothly. In the absence of training data,
modelling such correlation structures of real material density
functions is not possible. So we discretise the material sample
and treat values of the density function at chosen locations inside it
as independent and distribution-free parameters. Resolution of the
available image dictates the discretisation length of the model; three
models pertaining to distinct resolution classes are developed.  We
develop priors on the material density, such that these priors adapt
to the sparsity inherent in the density function. The likelihood is
defined in terms of the distance between the convolution of the
unknown functions and the image data. The posterior probability
density of the unknowns given the data is expressed using the
developed priors on the density and priors on the microscopy
correction function as elicitated from the Microscopy literature. We
achieve posterior samples using an adaptive Metropolis-within-Gibbs
inference scheme. The method is applied to learn the material density
of a 3-D sample of a nano-structure, using real image
data. Illustrations on simulated image data of alloy samples are also
included.

\vspace*{.3in}

\noindent\textsc{Keywords}: {{Bayesian methods},
{Inverse Problems},
{Nonparametric methods},
Priors on sparsity,
{Underdetermined problems}
}




\renewcommand{\baselinestretch}{1.5}


\section{Introduction}
\label{sec:intro}
\noindent
Nondestructive learning of the full three-dimensional material density
function in the bulk of an object, using available two dimensional
images of the object, is an example of a standard inverse problem
\cite{mayer,bereto,dorn_anisotropy,natterer,review_inv}. The image
results from the projection of the three dimensional material density
onto the image plane. However, inverting the image data does not lead
to a unique density function in general; in fact to render the
inversion unique, further measurements need to be invoked. For
instance, the angle at which the object is imaged or viewed is varied
to provide a series of images, thereby expanding available
information. This allows to achieve ditinguishability (or
identifiability) amongst the solutions for the material density.  A
real-life example of such a situation is presented by the pursuit of
material density function using noisy 2-dimensional (2-D) images taken
with electron microscopy techniques \ctp{panaretos}. Such non-invasive
and non-destructive 3-D density modelling of material samples is often
pursued to learn the structure of the material in its depth
\cite**{ndt}, with the ulterior aim of controlling the experimental
conditions under which material samples of desired qualities are
grown.

Formally, the projection of a density function onto a lower
dimensional image space is referred to as the Radon Transform; see
\ctn{radon}, \ctn{radon_xrays}. The inverse of this projection is also
defined but requires the viewing angle as an input and secondly,
involves taking the spatial derivative of the density function,
rendering the computation of the inverse projection numerically
unstable if the image data is not continuous or if the data comprises
limited-angle images \ctp{holder_conti,holder_conti_cite} or if noise
contaminates the data \ctp{lispeed}.
Furthermore, in absence of measurements of the viewing angle, the
implementation of this inverse is not directly
possible, as \ctn{panaretos} suggested. Even when the viewing angle is
in principle measurable, logistical difficulties in imaging at
multiple angles result in limited-angle images.
For example, in laboratory settings, such as when imaging with
Scanning Electron Microscopes (SEMs), the viewing angle is varied by
re-mounting the sample on stubs that are differently inclined each
time. This mounting and remounting of the sample is labour-intensive
and such action leads to the breaking of the vacuum within which
imaging needs to be undertaken, causing long waiting periods between
two consecutive images. When vacuum is restored and the next image is
ready to be taken, it is very difficult to identify the exact
fractional area of the sample surface that was imaged the last time,
in order to scan the beam over that very area. In fact, the
microscopist would prefer to opt for an imaging technique that allows
for imaging without needing to break the vacuum in the chamber at
all. Such logistical details are restrictive in that this allows
imaging at only a small number of viewing angles, which can cause the
3-D material density reconstruction to be numerically unstable. This
problem is all the more acute when the data is discontinuous and
heterogeneous in nature. Indeed, with the help of ingenious imaging
techniques such as compressive sensing \ctp{donoho06}, the requirement
of a large number of image data points is potentially mitigated. If
implemented, compressive imaging would imply averaging the image data
over a chosen few pixels, with respect to a known non-adaptive
kernel. Then the inversion of such compressed data would require the
computation of one more averaging or projection (and therefore one
more integration) than are otherwise relevant. While this is easily
done, the implementation of compressive sensing of electron microscopy
data can be made possible only after detailed instrumentational
additions to the imaging paraphernalia is made. Such instrumentational
additions are outside the scope of our work. We thus invoke a novel
design of imaging experiment that involves imaging at multiple values
of some relevant parameter that is easily varied in a continuous way
over a chosen range, unlike the viewing angle. In this paper, we
present such an imaging strategy that allows for multiple images to be
recorded, at each value of this parameter, when a cuboidal slab of a
material sampe is being imaged with an SEM.

In imaging with an SEM, the projection of the density function is
averaged over a 3-D region inside the material, the volume of which we
consider known; such a volume is indicated in the schematic diagram of
this imaging technique shown in Figure~\ref{fig:epma}.  Within this
region, a pre-fixed fraction of the atomistic interactions between the
material and the incident electron beam stay confined,
\cite{lee,goldstein}. The 2-D SEM images are taken in one of the
different types of radiation that are generated as a result of the
atomistic interactions between the molecules in the material and a
beam of electrons that is made incident on the sample, as part of the
experimental setup that characterises imaging with SEM,
\cite{goldstein,lee,reed}. Images taken with bulk electron microscopy
techniques potentially carry information about the structure of the
material sample under its surface, as distinguished from images
obtained by radiation that is reflected off the surface of the sample,
or is coming from near the surface\footnotemark.

\footnotetext{Since the radiation generated in
  an elemental three-dimensional volume inside the bulk of the
  material sample, is due to the interaction of the electron beam and
  the material density in that volume, it is assumed that the material
  density is proportional to the radiation density generated in an
  elemental 3-D volume. The radiation generated in an elemental volume
  inside the bulk of the material, is projected along the direction of
  the viewing angle chosen by the practitioner to produce the observed
  2-D image.}

In practice, the situation is rendered more complicated by the fact
that for the 2-D images to be formed, it is not just the material
density function, but the convolution of the material density function
with a kernel that is projected onto the image space. The nature of
the modulation introduced by this kernel is then also of interest, in
order to disentangle its effect from that of the unknown density
function that is responsible for generating the measured image. As for
the modulation, both enhancement and depletion of the generated
radiation is considered to occur in general. It is to be noted that
this kernel is not the measure of blurring of a point object in the
imaging technique, i.e. it is not the point spread function (PSF) of
the observed image, but is characteristic of the material at hand. The
kernel really represents the unknown microscopy correction function
that convolves with the material density function--this convolution
being projected to form the image data, which gets blurred owing to
specifics of the imaging apparataus. Learning the de-blurred image
from the PSF-convolved image data is an independent pursuit, referred
to as Blind Deconvolution \ctp{qui2008,hallqui2007,bookblind}. Here
our aim is to learn the material density and the microscopy correction
function using such noisy image data.

The convolution of the unknown kernel with the unknown density,
is projected first onto the 2-D surface of the system, and the
resulting projection is then successively marginalised over the 2 spatial
coordinates that span the plane of the system surface.
Thus, three successive projections result in the image data. 
Here we deal with the case of an image resulting from a composition of
a sequence of three independent projections of the convolution of the
material density function and the unknown kernel or microscopy
function. So only a sequence of inversions of the image allow us to
learn this density. Such multiple inversions are seldom addressed in
the literature, and here, we develop a Bayesian method that allows for
such successive inversions of the observed noisy, PSF-convolved 2-D
images, in order to learn the kernel as well as the density
function. The simultaneous pursuit of both the unknown density and the
kernel, is less often addressed than are reported attempts at density
reconstruction, under an assumed model for the kernel,
\cite**{goldstein,corr_epm,heinrich,pichou_85}.

We focus here on the learning of a density function that is not
necessarily convex, is either sparse or dense, is often multimodal -
with the modes characterised by individualised, discontinuous
substructure and abrupt bounds, resulting in the material density
function being marked by isolated islands of overdensity and sharp
density contrasts; we shall see this below in examples of
reconstructed density functions presented in
Section~\ref{sec:illustrations} and Section~\ref{sec:real}. Neither
the isolated modes of the material density function in real-life
material samples nor the adjacent bands of material over-density, can
be satisfactorily modelled with a mixture model. Again, modelling of
such a real-life trivariate density function with a high-dimensional
Gaussian Process (GP) is not possible in lieu of training data, given
that the covariance function of the invoked GP will then have to be
non-stationary \ctp{neal98} and exhibit non-smooth spatial
variation. Here by ``training data'' is implied data comprising a set
of known values of the material density, at chosen points inside the
sample; such known values of the density function are unknown.
Even if the covariance function were varying smoothly over the 3-D
spatial domain its parametrisation is $ad hoc$ in lieu of training
data but with an abruptly evolving covariance function--as in the
problem of inverting SEM image data--such parametrisation becomes
impossible \ctp{paciorek}, especially when there is no training data
available.
At the same time, the blind implementation of the
inverse Radon Transform would render the solution unstable given that
the distribution of the image data in real-life systems is typically,
highly discontinuous.




In the following section, (Section~\ref{sec:application}) we describe
the experimental setup delineating the problem of inversion of the 2-D
images taken with Electron Microscopy imaging techniques, with the aim
of learning the sub-surface material density and the kernel. In
addition, we present the novel design of imaging experiment that
achieves multiple images. In Section~\ref{sec:novelty} we discuss the
outstanding difficulties of multiple and successive inversions and
qualitatively explain the relevance of the advanced solution to this
problem. This includes a discussion of the integral representation of
the multiple projections (Section~\ref{sec:exptl}), the outline of the
Bayesian inferential scheme adopted here (Section~\ref{sec:met}) and a
section on the treatment of the low-rank component of the unknown
material density (section~\ref{sec:lowrank}).  The details of the
microscopy image data that affect the modelling at hand are presented
in the subsequent section. This includes a description of the 3 models
developed to deal with the three classes of image resolution typically
encountered in electron microscopy (Section~\ref{sec:cases_new}) and
measurement uncertainties of the collected image data
(Section~\ref{sec:data}). Two models on the kernel, as motivated by
small and high levels of information available on the kernel given the
material at hand are discussed in Section~\ref{sec:correction}; priors
on the parameters of the respective models is discussed here as
well. Section~\ref{sec:sparse} is devoted to the development of priors
on the unknown density function such that the priors are adaptive to
the sparsity in the density. In Section~\ref{sec:model}, the
discretised form of the successive projections of the convolution of
the unknown functions is given for the three different models.
Inference is discussed in
Section~\ref{sec:inference}. Section~\ref{sec:posterior} discusses
with salient aspects of the uniqueness of the solutions. Application
of the method to the analysis of real Scanning Electron Microscope
image data is included in Section~\ref{sec:real} while results from
the inversion of simulated image data are presented in
Section~\ref{sec:illustrations}. Relevant aspects of the methodology
and possible other real-life applications in the physical sciences are
discussed in Section~\ref{sec:discussions}.

\section{Application to Electron Microscopy Image Data}
\label{sec:application}
\noindent
For our application, the system i.e. the slab of material, is modelled as a
rectangular cuboid such that the surface at which the electron beam is
incident is the $Z$=0 plane. Thus, the surface of the system is
spanned by the orthogonal $X$ and $Y$-axes, each of which
is also orthogonal to the $Z$-axis that spans the depth of the slab. In our
problem, the unknown functions are the material density $\rho(x,y,z)$
and the kernel $\eta(x,y,z)$.

Here we learn the unknown functions by inverting the image data
obtained with a Scanning Electron Microscope. Though the learning of
the convolution $\rho\ast\eta$ of $\rho(x,y,z)$ and $\eta(x,y,z)$, is
in principle ill-posed, we suggest a novel design of imaging
experiment that allows for an expansion of the available information
leading to $\rho\ast\eta$ being learnt uniquely when noise in the data
is small (Section~\ref{sec:posterior}). This is achieved by imaging
each part of the material sample at $N_{{\textrm{eng}}}$ different values of 
a parameter $E$ such that images taken at different values $\epsilon$ of
this parameter $E$, carry information about the density function from
different depths under the sample surface. Such is possible if $E$
represents the energy of the electrons in the incident electron beam
that is used to take the SEM image of the system; since the
sub-surface penetration depth of the electron beam increases with
electron energy, images taken with beams of different energies
inherently bear information about the structure of the material at
different depths.

The imaging instrument does not image the whole sample all at the same
time. Rather, the imaging technique that is relevant to this problem
is characterised by imaging different parts of the sample,
successively in time. At each of these discrete time points, the
$i$-th part, ($i=1,\ldots,N_{{\textrm{data}}}$) of the sample is viewed by the
imaging instrument, along the viewing vector ${\bv}_i$ to record the
image datum in the $i$-th pixel on the imaging screen. Thus, every
pixel of image data represents information about a part of the
sample. The $i$-th viewing vector corresponds to the $i$-th point of
incidence of the electron beam on the surface of the material
sample. The image datum recorded in the $i$-th pixel then results from
this viewing, and harbours information about the sample structure
inside the $i$-th interaction-volume, which is the region that bounds
atomistic interactions between the incident electron beam and
molecules of the material (see Figure~\ref{fig:epma}). The point of
incidence of the $i$-th viewing vector, i.e. the centre of the $i$-th
interaction volume is marked in the diagram. The volume of any such
3-D interaction-volume is known from microscopy theory \ctp{kanaya}
and is a function of the energy $E$ of the beam electrons. In fact,
motivated by the microscopy literature, \ctp{kanaya}, we model the
shape of the $i$-th interaction-volume as hemispherical, centred at
the incidence point of $\bv_i$, with the radius of this hemisphere
modelled as $\propto E^{1.67}$. We image each part of the sample at
$N_{{\textrm{eng}}}$ different values of $E$, such that the $k$-th value of $E$
is $\epsilon_k$, $k=1,\ldots,N_{{\textrm{eng}}}$. To summarise, the data comprise
a $N_{{\textrm{eng}}}$ 2-D images, each image being a square spatial
array of ${N_{{\textrm{data}}}}$ number of pixels such that at the $i$-th pixel, for the
$k$-th value of $E$, the image data is $\tilde{I}_i^{(k)}$ where
$i=1,\ldots,N_{{\textrm{data}}}$. Consideration of this full set of images will
then suggest the sub-surface density function of the sample, in a
fully discretised model. Convolution $\rho\ast\eta$ of $\rho(x,y,z)$
and $\eta(x,y,z)$ is projected onto the system surface, i.e. the $Z$=0
plane and this projection is then further projected onto one of the
$X$ or $Y$ axes, with the resulting projection being projected once
again, to the central point of the interaction-volume created by the
$i$-th incidence of the beam of energy $\epsilon_k$, to give rise to
the image datum in the $i$-th pixel in the $k$-th image.

Expanding information by imaging at multiple values of $E$ is
preferred to imaging at multiple viewing angles for logistical reason
as discussed above in Section~\ref{sec:intro}. Further, the shape of
the interaction-volume is rendered asymmetric about the line of
incidence of the beam when the beam is tilted to the vertical
\ctp[Figure 3.20 on page 92 of]{polymer}, where the asymmetry is
dependent on the tilt angle and the material. Then it is no longer
possible to have confidence in any parametric model of the geometry of
the interaction-volume as given in microscopy theory. Knowledge of the
symmetry of the interaction-volume is of crucial importance to any
attempt in inverting SEM image data with the aim of learning the
material density function. So in this application, varying the tilt
angle will have to be accompanied by approximations in the modelling
of the projection of $\rho\ast\eta$. We wish to avoid this and
therefore resort to imaging at varying values of $E$.

Physically speaking, the unknown density $\rho(x,y,z)$ would be the
material density representing amount or mass of material per unit 3-D
volume. The measured image datum could be the radiation in Back
Scattered Electrons or X-rays, in each 2-D pixel, as measured by a
Scanning Electron Microscope.

\begin{figure}[!t]
{$\begin{array}{c c}                                                          
   \includegraphics[height=.2\textheight]{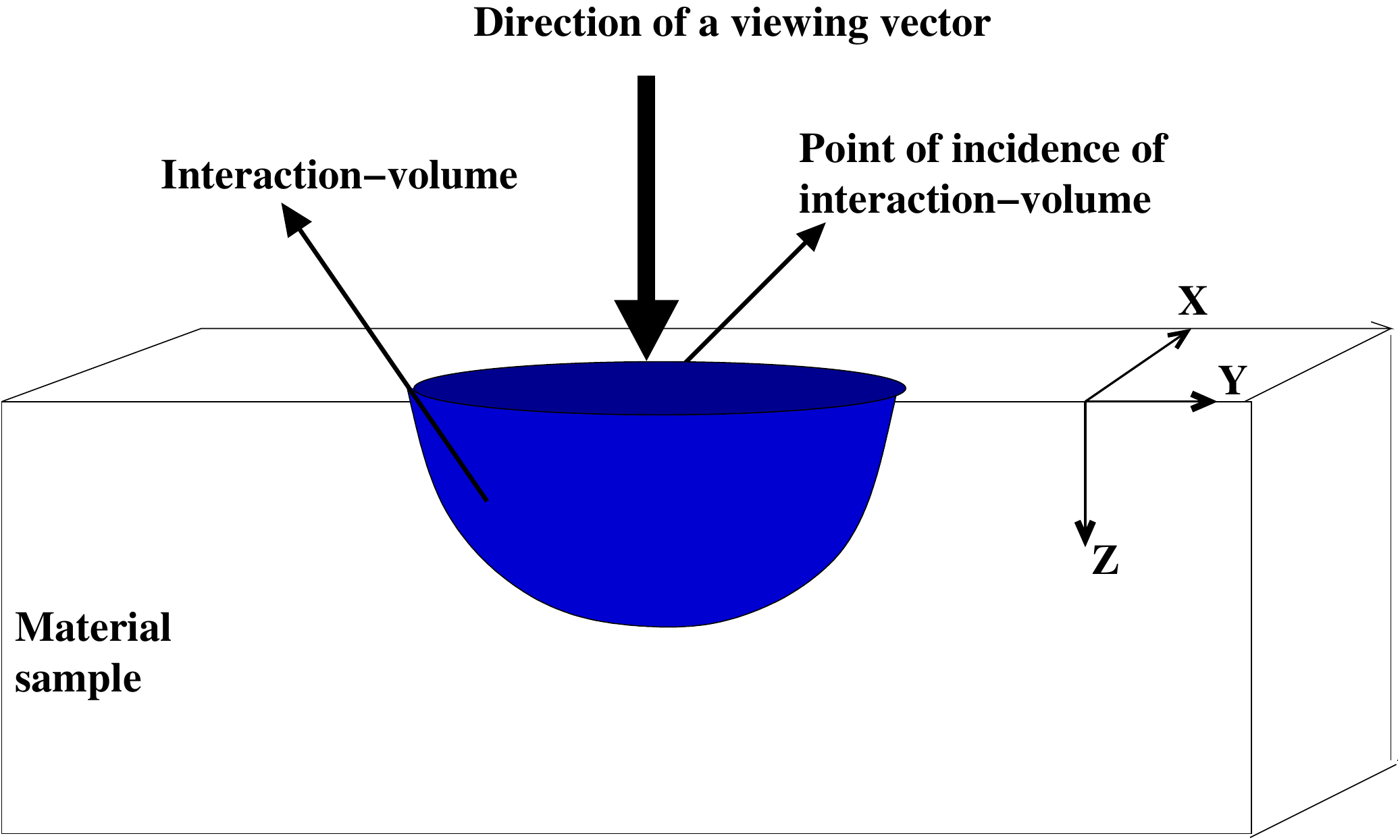}
  \includegraphics[height=.22\textheight]{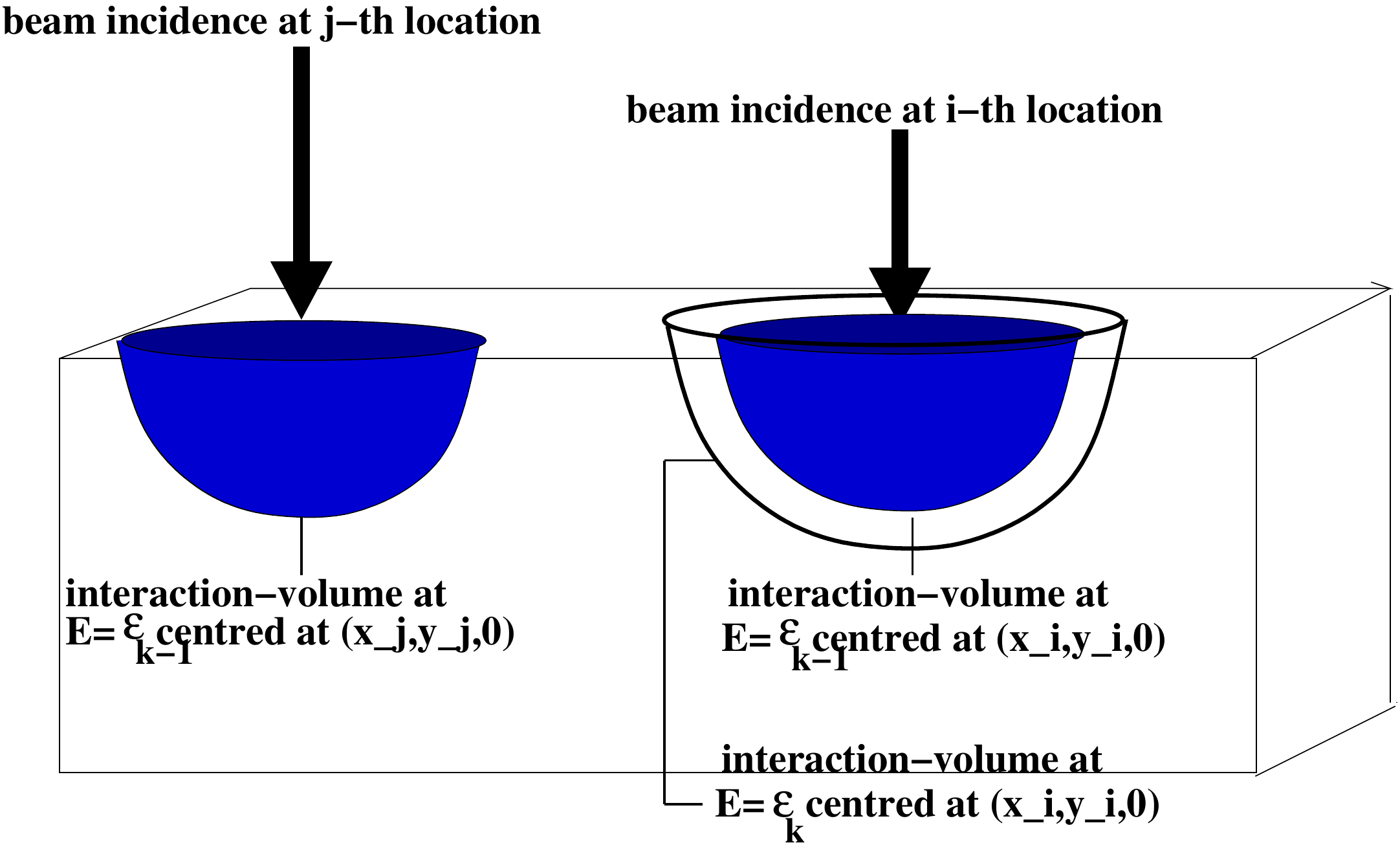}  
\end{array}$}
\caption{{\small{{\it Left:} schematic diagram of the imaging
      experiment relevant to the application of inverting image data
      taken with a Scanning Electron Microscope, to learn sub-surface
      material density and the kernel. The interaction-volume created
      by the incidence of the electron beam (shown in thick arrow) is
      modelled as a hemisphere (shown in blue) centred at the point of
      incidence of this beam with radius given as a function of the
      energy $E$ of the beam. At each of the $N_{{\textrm{data}}}$ incidences of
      the beam, $N_{{\textrm{eng}}}$ images are taken with beams that are
      distinguished by the energy $E$ of the electrons in them. The
      $X$, $Y$ and $Z$-axes characterising the 3-dimensional grid used
      in the modelling are marked. The $i$-th incidence of the beam is
      on the point $(x_i,y_i,0)$ on the sample surface. {\it Right:}
      The $j$-th and $i$-th beam incidences are depicted. At the
      $i$-th incidence, the interaction-volumes resulting from
      interactions of the material with beams of energy $E=\epsilon_k$
      (larger interaction-volume in outline only) and
      $E=\epsilon_{k-1}$ (smaller interaction-volume in blue) are
      schematically depicted.}}}
\label{fig:epma}
\end{figure}

\section{Modelling}
\label{sec:novelty}
\noindent
The general inverse problem is $I = {\cal P}(\rho) + \varepsilon$
where the data $I:{\mathbb
  R}^m\longrightarrow{\cal D}\subseteq{\mathbb R}$, while the unknown function
$\rho:{\mathbb R}^n\longrightarrow{\cal H}\subseteq{\mathbb
  R}$, with $m\leq n$. In particular, 3-D modelling involves the case of $n=3$,
$m=2$. The case of
$n > m$ is fundamentally an ill-posed problem
\ctp{tricomi,tarantola}. Here $\varepsilon$ is the measurement noise,
the distribution of which, we assume known.  

In our application, the data is represented as the projection of the
convolution of the unknown functions as:
\begin{equation}
{\tilde{I}}={\cal C}[\rho\ast\eta] + \varepsilon, \nonumber
\end{equation}
where the projection operator ${\cal C}$ is a contractive projection from
the space that $\rho\ast\eta$ lives in - which is $\subseteq{\mathbb
  R}^3$ - onto the image space ${\cal D}$. ${\cal C}$ itself is a
  composition of 3 independent projections in general,
\begin{equation}
{\cal C}=P_1\circ P_2\circ P_3, \nonumber
\end{equation} 
where $P_1$ is the projection onto the $Z$=0 plane, followed by $P_2$
- the projection onto the $Y$=0 axis, followed by $P_3$ - projection
onto the centre of a known material-dependent three dimensional region
inside the system, namely the {\it interaction volume}. These
3 projections are commutable, resulting in an effective contractive
projection of $\rho\ast\eta$ onto the centre of this interaction
volume. Thus, the learning of $\rho\ast\eta$ requires multiple (three)
inversions of the image data. In this sense, this is a harder than
usual inversion problem. The interaction volume and its centre - at
which point the electron beam is incident - are shown in
Figure~\ref{fig:epma}.

It is possible to reduce the learning of $\rho\ast\eta$ to a
least-squares problem in the low-noise limit, rendering the inverse
learning of $\rho\ast\eta$ unique by invoking the Moore Penrose
inverse of the matrix ${\bf C}$ that is the matrix representation of
the ${\cal C}$ operator (Section~\ref{sec:posterior}). The learning of
$\rho({x,y,z})$ and $\eta({x,y,z})$ individually, from the uniquely
learnt $\rho\ast\eta$ is still an ill-posed problem. In a
non-Bayesian framework, a comparison of the number of unknowns with
the number of measured parameters is relevant; imaging at
$N_{{\textrm{eng}}}$ number of values of $E$ suggests that $N_{{\textrm{data}}}\times
N_{{\textrm{eng}}}$ parameters are known while at most $N_{{\textrm{data}}}\times N_{{\textrm{eng}}} +
N_{{\textrm{eng}}}$ are unknown. Thus, for the typical values of $N_{{\textrm{eng}}}$ and
$N_{{\textrm{data}}}$ used in applications, ratio of known to unknown parameters
is $\geq 0.99$, using the aforementioned design of experiment (see
Section~\ref{sec:quant}). The unknown $N_{{\textrm{eng}}}$ parameters still
renders the individual learning of $\rho({x,y,z})$ and $\eta({x,y,z})$
non-unique. Such considerations are however relevant only in the
absence of Bayesian spatial regularisation
\ctp{andrewreview,cotter2010}, and as we will soon see, the lack of
smooth variation in the spatial correlation structure underlying
typically discontinuous real-life material density functions, render
the undertaking of such regularisation risky.

However in the Bayesian framework we seek the posterior probability of
the unknowns given the image data. Thus, in this approach there is no
need to directly invert the sequential projection operator ${\cal C}$;
the variance of the posterior probability density of the unknowns
given the data is crucially controlled by the strength of the priors
on the unknowns \ctp{gouveia,bayesinv}.
We develop the priors on the unknowns using as much of information as
is possible.  Priors on the kernel can be weak or strong depending on
information available in the literature relevant to the imaged
system. Thus, in lieu of strong priors that can be elicited from the
literature, a distribution-free model for $\eta({x,y,z})$ is
motivated, characterised by weak priors on the shape of this unknown
function. On the contrary, if the shape is better known in the
literature for the material at hand, the case is illustrated by
considering a parametric model for $\eta({x,y,z})$, in which priors
are imposed on the parameters governing this chosen shape. Also, the
material density function can be sparse or dense. To help with this, we
need to develop a prior structure on $\rho({x,y,z})$ that adapts to
the sparsity of the density in its native space. It is in this
context, that such a problem of multiple inversions is well addressed
in the Bayesian framework.

\subsection{Defining a general voxel}
\noindent
Lastly, we realise that the data at hand is intrinsically discrete,
owing to the imaging mechanism. Then,
\begin{itemize}
\item for a given beam energy, collected image data is characterised
  by a resolution that depends on the size of the electron beam that
  is used to take the image, as well as the details of the particular
  instrument that is employed for the measurements. The SEM cannot
  ``see'' variation in structure at lengths smaller than the beam
  size. In practice, resolution $\omega$ is typically less than the
  beam size; $\omega$ is given by the microscopist. Then only one
  image datum is available over an interval of width $\omega$ along
  each of the $X$ and $Y$-axes, i.e. only one datum is available from
  a square of edge $\omega$, located at the beam pointing. This
  implies that using such data, no variation in the material density
  can be estimated within any square (lying on the $X-Y$ plane) of
  side $\omega$, located at a beam pointing,
  i.e. $\rho(x,y,z)=\rho(x+\delta_1,y+\delta_2,z)$, where
  $\delta_1,\delta_2\in[0,\omega)$, $\forall\:x,y,z$. The reason for not being able to {\it predict} the density over lengths smaller than $\omega$, in typical real-life material samples, is discussed below.
\item image data are recorded at discrete (chosen) values of the
  beam energy $E$ such that for a given beam pointing, a beam at a
  given value $\epsilon_k$ of $E$ carries information about the
  material density from within the sub-surface depth of
  $h^{(k)}$. Then 2 beams at successive values $\epsilon_{k-1}$ and
  $\epsilon_k$ of $E$ allow for estimation of the material density
  within the depth interval $[h^{(k-1)}, h^{(k)})$ that is bounded by
    the beam penetration depths. Here $k=1,2,\ldots,N_{{\textrm{eng}}}$,
    $h^{(0)}$=0. This implies that for a given beam pointing, over any
    such interval, no variation in the material density can be learnt
    from the data, i.e. $\rho(x,y,z)=\rho(x,y,z+\delta_3)$ where
    $\delta_3\in[h^{(k-1)}, h^{(k)})$, $\forall\:x,y,z$. Again, the
      reason for inability to predict the density at any $z$ is
      explained below.
\end{itemize}
As is evident in these 2 bulletted points above, we are attempting to
learn the density at discrete points identified by the resolution in
the data. The inability to predict the density at any point inside the
sample could be possible in alternative models, for example when a
Gaussian Process (GP) prior is used to model the trivariate density
function that is typically marked by disjoint support in ${\mathbb
  R}^3$ and/or by sharp density contrasts (seen in a real density as
in Figure~\ref{fig:real_den} and those learnt from simulated image
data, as in Figure~\ref{fig:colour}). While the covariance structure
of such a GP would be non-stationary (as motivated in the introductory
section), the highly discontinuous nature of real-life density
functions would compel this non-stationary covariance function to not
vary smoothly. Albeit difficult, the parametrisation of such
covariance kernels can in principle be modelled using a training
data--except for the unavailability of such training data (comprising
known values of density at chosen points inside the given 3-D
sample). As training data is not at hand, parametrisation of the
non-stationary covariance structure is not possible, leading us to
resort to a fully discretised model in which we learn the density
function at points inside the sample, as determined by the data
resolution. In fact, the learnt density function could then be used as
training data in a GP-based modelling scheme to help learn the
covariance structure.  Such an attempt is retained for future
consideration.

Thus, prediction would involve (Bayesian) spatial regularisation which
demands modelling of the abruptly varying spatial correlation
structure that underlies the highly discontinuous, trivariate density
function of real life material samples (see
Figure~\ref{fig:real_den}). The spatial variation in the correlation
structure of the image data cannot be used as a proxy for that of the
material density function, since the former is not a pointer to the
latter, given that compression of the density information within a 3-D
interaction-volume generates an image datum. Only with training data
comprising known material density at chosen points, can we model the
correlation structure of the density function. In lieu of
training data, as in the current case, modelling of such correlation
is likely to fail.

Hence we learn the density function at the grid points of a 3-D
Cartesian grid set up in the space of $X$, $Y$ and $Z$. No variation
in the material density inside the $ik$-th grid-cell can be learnt,
where this grid-cell is defined as the ``$ik$-th voxel'' which is the
cuboid given by
\begin{enumerate}
\item[--]square cross-sectional area (on a $Z=$constant plane) of size $\omega^2$, with edges parallel to the $X$ and $Y$ axes, located at the $i$-th beam incidence,
\item[--]depth along the $Z$-axis ranging from $h^{(k-1)}\leq z< h^{(k)}$,
\end{enumerate}
with $i=1,2,\ldots,N_{{\textrm{data}}}$, $k=1,2,\ldots,N_{{\textrm{eng}}}$, $h^{(0)}=0$. Then
one pair of opposite vertices of the $ik$-th voxel are respectively at points
$(x_i,y_i,h^{(k-1)})$ and $(x_i,y_i,h^{(k)})$, where the $i$-th beam incidence is at the point $(x_i,y_i,0)$ and
\begin{equation}
{x_i:= i\:{\textrm {modulo}}\:(\sqrt{N_{{\textrm{data}}}})\quad
y_i:={\textrm{int}}(i/\sqrt{N_{{\textrm{data}}}})+1.}
\label{eqn:modulo}.  
\end{equation} 
See Figure~\ref{fig:ikth} for a schematic depiction of a general
voxel.  Here $\sqrt{N_{{\textrm{data}}}}$ is the number of intervals of width
$\omega$ along the $X$-axis, as well as along the $Y$-axis. Thus, on
the part of the $X-Y$ plane that is imaged, there are $N_{{\textrm{data}}}$
squares of size $\omega$. Formally, the density function that we can
learn using such discrete image data is represented as
\begin{equation}
\xi^{(k)}_{i} := \rho(x,y,z)\quad\mbox{for}\quad x\in[x_i,x_i+\omega),y\in[y_i,y_i+\omega), z\in[h^{(k-1)},h^{(k)}),
\end{equation}

This representation suggests that a discretised version of the sought
density function is what we can learn using the available image
data. We do not invoke any distribution for the unknown parameters\\
$\xi^{(1)}_1,
\xi^{(2)}_1,\ldots,\xi^{(N_{{\textrm{eng}}})}_1,\xi^{(1)}_2,\ldots,\xi^{(N_{{\textrm{eng}}})}_2,\ldots,\xi^{(1)}_{N_{{\textrm{data}}}},\ldots,\xi^{(N_{{\textrm{eng}}})}_{N_{{\textrm{data}}}}$. It
is to be noted that the number of parameters that we seek is {\it deterministically}
driven by the image data that we work with--the discretisation of the
data drives the discretisation of the space of $X$, $Y$ and $Z$ and
the density in each of the resulting voxels is sought. Thus, when the
image has a resolution $\omega$ and there are $N_{{\textrm{eng}}}$
such image data sets recorded at a value of $E$,
the number of sought density parameters is $N_{{\textrm{data}}}\times N_{{\textrm{eng}}}$
where $N_{{\textrm{data}}}$ is the number of squares of size $\omega$ that
comprise part of the surface of the material slab that is imaged. Here
the number of parameters $N_{{\textrm{data}}}\times N_{{\textrm{eng}}}$ is typically large;
in the application to real data that we present later, this number is
2250. We set up a generic model that allows for the learning of such
large, data-driven number of distribution-free density parameters in a
Bayesian framework.

\begin{figure}[!t]
\vspace*{-3in}
     \begin{center}
  {  \includegraphics[width=10cm]{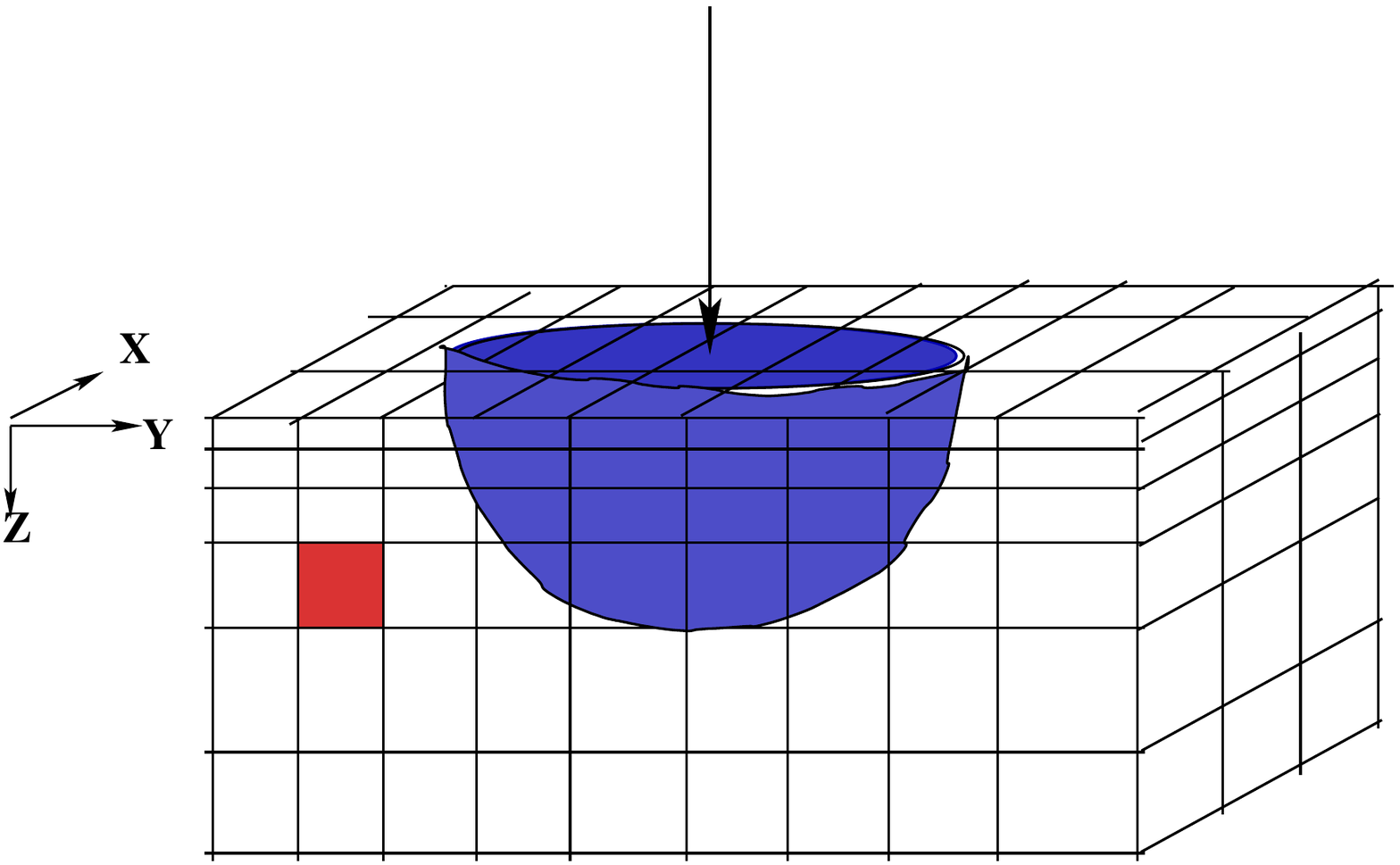}}
     \end{center}
\caption{\small{Figure to bring out the general ($ik$-th) voxel, as
    distinguished from the general ($ik$-th) interaction-volume. For
    the value $\epsilon_k$ of the imaging parameter $E$, the $i$-th
    viewing vector $\bv_i$ impinges the $Z$=0 surface at the material
    slab at point $(x_i,y_i,0)$, creating the $ik$-th
    interaction-volume (in blue) that is modelled as a hemisphere of
    radius $R0^{(k)}$. The maximal depth that this interaction-volume
    extends to, is $h^{(k)}$. Thus, under the assumption of
    hemi-spherical shape, $R0^{(k)}=h^{(k)}$. Here,
    $i=1,\ldots,N_{{\textrm{data}}}$ and $k=1,\ldots,N_{{\textrm{eng}}}$. In the figure,
    the beam is incident on the 14-th square on the surface, where the origin is at the marked vertex O. Also, the depicted interaction-volume goes down to the 4-th $Z$-bin, i.e. has a radius=$R0^{(4)}=h^{(4)}$. Then the depicted interaction-volume is the $ik$-th one, with $i=14$ and
    $k$=4. The example voxel the outermost face of which on the $X=0$
    plane is shown in red, lies between depths $h^{(3)}$ and
    $h^{(4)}$, and has $i=2$; this is then the $ik$-th voxel, with
    $i$=2, $k$=4. The grid-cells on the $Z$=constant planes are
    squares of side $\omega$, the resolution length of the imaging
    instrumentation. The gridding along the $Z$-axis is
    non-uniform. Material density inside the $ik$-th voxel is the
    constant $\xi_i^{(k)}$ and for all $i=1,\ldots,N_{{\textrm{data}}}$,
    correction function in this voxel is $\eta^{(k)}$,
    $k=1,\ldots,N_{{\textrm{eng}}}$.}}
\label{fig:ikth}
\end{figure}


\subsection{Defining a general interaction-volume}
\label{sec:exptl}
\noindent
Under the consideration of a hemi-spherical shape of the
interaction-volume, the $ik$-th interaction-volume is completely
specified by pinning its centre to the $i$-th beam pointing
at $(x_i,y_i,0)$ on the $Z$=0 plane and fixing its radius to
$R0^{(k)}\propto{\epsilon_k}^{1.67}$, \ctp{kanaya} where the constant of
proportionality is known from microscopy theory\footnotemark and $x_i$
and $y_i$ are defined in Equation~\ref{eqn:modulo}. In this
hemispherical geometry, the maximal depth that the $ik$-th interaction
volume goes down to is $h^{(k)}=R0^{(k)}$.

\footnotetext{
In the context of the application to microscopy, atomic theory studies
undertaken by \citeasnoun{kanaya} suggest that the maximal depth to
which an electron beam can penetrate inside a given material sample is  
$h^{(k)}$ (measured in $\mu$m), for a beam of electrons of energy $E=\epsilon_k$ (measured in kV), inside material of mass density of
$d$ (measured in gm cm$^{-3}$), atomic number ${\cal Z}$ and
atomic weight $A$ (per gm per mole), is
\begin{equation}
\label{eqn:kanaya}
h^{(k)} = \displaystyle{\frac{0.0276 A \epsilon_k^{1.67}}{d {\cal Z}^{0.89}}}\quad k\in{\mathbb Z}_{+}.
\end{equation}
Here ${\cal Z}$ is an integer valued atomic number of the material,
${\cal Z} > 0$ while $A$ and $d$ are positive-definite real valued
constants. As $E$ increases, the depth and radial extent of the
interaction-volume increases.}

The measured image datum in the $i$-th pixel, created at
$E=\epsilon_k$, is ${\tilde I}_i^{(k)}$. This results when the
convolution $(\rho\ast\eta)_i^{(k)}$, of the unknown density and
kernel in the voxels that lie inside the $ik$-th
interaction-volume, is sequentially projected (as mentioned
in Section~\ref{sec:novelty}), onto the centre of
the $ik$-th interaction-volume, i.e. onto the point $(x_i,y_i,0)$. 
This projection is referred to as ${\cal C}(\rho\ast\eta)_i^{(k)}$. 
The integral representation of such a projection suggests 3
integrals, each over the three spatial coordinates that define the
$ik$-th interaction-volume, according to
\begin{equation}
\label{eqn:general}
{\cal C}(\rho\ast\eta)_i^{(k)} = \displaystyle{\frac{\int_{R=0}^{R0^{(k)}}\int_{\theta=0}^{\theta_{max}} RdR d\theta\int_{z=0}^{z=z_{max}(R,\theta)}\rho({\bf s})\ast\eta({\bf s})dz}
{ \int_{R=0}^{R_{max}}\int_{\theta=0}^{\theta_{max}} RdR d\theta }}  
\end{equation}
for $i=1,\ldots,N_{{\textrm{data}}}$. Here the vector ${\bf s}$
represents value of displacement from the point of incidence
$(x_i,y_i,0)$, to a general point $(x,y,z)$, inside the
interaction-volume, i.e.
\begin{eqnarray}
{\bf s}:= (x-x_i, y-y_i, z)^T &=& (R\cos\theta, R\sin\theta, z)^T\nonumber \\
R = \sqrt{(x-x_i)^2+(y-y_i)^2} &&\quad \tan\theta = \displaystyle{\frac{y-y_i}{x-x_i}}. \nonumber \\ 
\end{eqnarray}




\subsection{Cases classed by image resolution}
\label{sec:cases_new}
\noindent
We realise that the thickness of the electron beam imposes an
upper limit on the image resolution, i.e. on the smallest length
$\omega$ over which structure can be learnt in the available image
data. For example, the resolution is finer when the SEM image is taken
in Back Scattered Electrons ($\omega\lesssim$ 0.01$\mu$m) than in
X-rays ($\omega\sim$ 1$\mu$m). The comparison between the cross-sectional
areas of a voxel and of an interaction-volume, on the $Z$=constant plane 
is determined by
\begin{enumerate}
\item[--] $\omega$ since the area of the voxel at any $z$ is $\omega^2$.
\item[--] the atomic-number parameter ${\cal Z}$ of the material at
  hand (see Equation~\ref{eqn:kanaya}) since the radius $R0^{(k)}$ of the
  hemi-spherical interaction-volume is a monotonically decreasing
  function of ${\cal Z}$, at $E=\epsilon_k$.
\end{enumerate} 
Then we can think of 3 different
resolution-driven models such that
\begin{enumerate}
\item[1st~model:] voxel cross-sectional area on the $Z$=0 plane exceeds that of interaction-volumes attained at all $E$, i.e. \\
$\pi(R0^{(k)})^2\leq \omega^2,\:\forall\:k=1,\ldots,N_{{\textrm{eng}}}$.
\item[2nd~model:] voxel cross-sectional area on the $Z$=0 plane exceeds that of interaction-volumes at some values of $E$ but is lower than that of interaction-volumes attained at higher $E$ values, i.e.
$\pi(R0^{(k)})^2\leq \omega^2,\:\forall\:k=1,\ldots,k_{in}$ and
$\pi(R0^{(k)})^2 > \omega^2,\:\forall\:k=k_{in}+1,\ldots,N_{{\textrm{eng}}}$. 
\item[3rd~model:] voxel cross-sectional area on the $Z$=0 plane is lower than that of interaction-volumes attained at all $E$, i.e.
$\pi(R0^{(k)})^2 > \omega^2,\:\forall\:k=1,\ldots,N_{{\textrm{eng}}}$.
\end{enumerate}
The first two models are realised for coarse resolution of the imaging
technique, i.e. for large $\omega$. This is contrasted with the
3rd~model, which pertains to fine resolution, i.e. low $\omega$. Since
$R0^{(k)}$ is a monotonically decreasing function of ${\cal Z}$
(Equation~\ref{eqn:kanaya}), the 1st~model is feasible for
``high-${\cal Z}$ materials'', while the 2nd~model is of relevance when
the material at hand is a ``low-${\cal Z}$ material''. To make this
more formal, we introduce two classes of materials as follows. For fixed
values of all parameters but ${\cal Z}$,
\begin{eqnarray}
\label{eqn:highzlowz}
{\textrm{if}}\quad\omega^2 \geq \pi [R0^{(N_{{\textrm{eng}}})}]^2\vert {\cal Z}, & & \quad{\textrm{``high-${\cal Z}$'' material}}\\
{\textrm{if}}\quad\omega^2 < \pi [R0^{(N_{{\textrm{eng}}})}]^2\vert {\cal Z}, & & \quad{\textrm{``
low-${\cal Z}$'' material}}.\nonumber  
\end{eqnarray}

The computation of the sequential projections involve multiple
integrals (see Equation~\ref{eqn:general}). The required number of
such integrals can be reduced by invoking symmetries in the material
density in any interaction-volume.
Such symmetries become avaliable, as the resolution in the imaging
technique becomes coarser. For example, for the 1st model, when the
resolution is the coarsest, the minimum length $\omega$ learnt in the
data exceeds the diameter of the lateral cross-section of the largest
interaction-volume achieved at the highest value of $E$, i.e for
$E=\epsilon_{N_{{\textrm{eng}}}}$. (Here ``lateral cross section'' implies the
cross-section on the $Z$=0 plane). Then for the coarsest resolution,
$2R0^{(N_{{\textrm{eng}}})} \leq \omega$ (see Figure~\ref{fig:3cases}) and so,
$\displaystyle{\pi[R0^{(N_{{\textrm{eng}}})}]^2} \leq \omega^2$ which means the
lateral cross-section of the $ik$-th interaction-volume is contained
inside that of the $ik$-th voxel. Now in our discretised
representation, the material density inside any voxel is a
constant. Therefore it follows that when resolution is the coarsest,
material density at any depth inside an interaction-volume, is a
constant independent of the angular coordinate $\theta$ (see
Equation~\ref{eqn:general}). Thus, the projection of $\rho\ast\eta$
onto the centre of the interaction-volume, does not involve
integration over this angular coordinate.

However, such will not be possible for the 3rd model at finer imaging
resolution. When the resolution is the finest possible with an SEM,
$\omega$ is the smallest. In this case, the lateral cross-section of
many voxels fill up the cross-sectional area of a given interaction
volume. Thus, the density is varying within the
interaction-volume. Thus, the projection of $\rho\ast\eta$ cannot
forgo integrating over $\theta$.

The last case that we consider falls in between these
extremes of resolution; this is the 2nd model. Isotropy is valid
$\forall\:k=1,\ldots,k_{in}$ but not for $k=k_{in}+1,\ldots,N_{{\textrm{eng}}}$;
(see middle panel in Figure~\ref{fig:3cases}).

\begin{figure}[!t]
\vspace*{-3in}
     \begin{center}
  {$\begin{array}{c c c}
       \includegraphics[width=7cm]{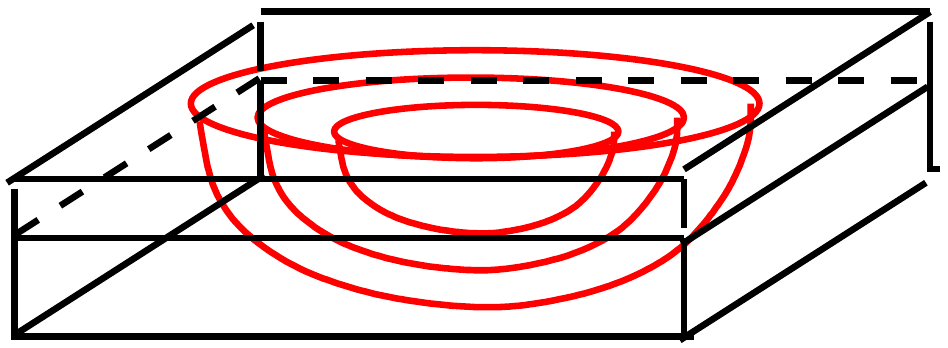}
       \includegraphics[width=7cm]{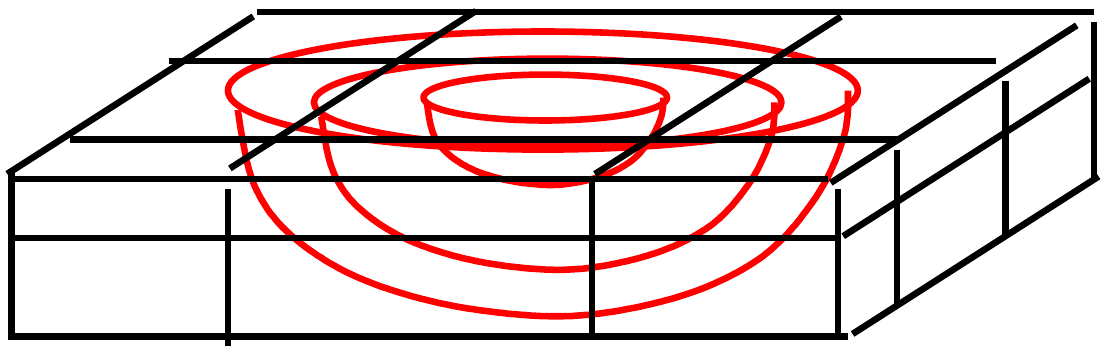}
       \includegraphics[width=7cm]{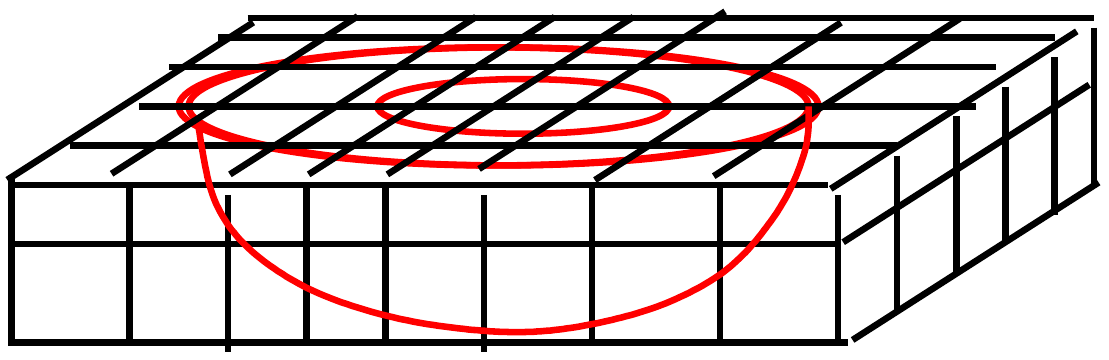}
  \end{array}$}
     \end{center}
\caption{\small{A schematic diagram to depict the three models that we
    work with. When the instrumental resolution is the coarsest out of
    the three cases, i.e. $\omega$ is the largest, (for high-${\cal
      Z}$ materials), the size of the $ik$-th interaction-volumes 
    (depicted in red) is such that their cross-sectional
    areas on a $Z$=constant plane fall short of the area of a voxel
    (outlined in black), on this plane. Here $k=1,\ldots,N_{{\textrm{eng}}}$
    though in the representation on the left, a value of 2 has been
    used for $N_{{\textrm{eng}}}$. Thus, the projection ${\cal
      C}(\rho\ast\eta)_i^{(k)}$ onto the centre of an
    interaction-volume is done by considering an isotropic density
    function, since density inside a voxel is a constant. {\it Middle
      panel:} (for low-${\cal Z}$ materials) the cross-sectional area
    $\pi(R0^{(k)})^2$ of the $ik$-th interaction-volume and $\omega^2$
    of any voxel, on the $Z$=0 plane are such that
    $\pi(R0^{(k)})^2\leq\omega^2$ for $k=1,\ldots,k_{in}$. On the
    other hand, for $k= k_{in}+1,\ldots,N_{{\textrm{eng}}}$,
    $\pi(R0^{(k)})^2>\omega^2$ and therefore for such $k$, the $ik$-th
    interaction-volume spills into the neighbouring voxels on the
    $Z$=0 plane. Thus, the contribution to projection ${\cal
      C}(\rho\ast\eta)_i^{(k)}$ from the $ik$-th
    interaction-volume includes contribution from these neighbouring
    voxels as well. Such contribution is modelled using a
    nearest-neighbour weighted averaging. {\it Left:} for the finest
    instrumental resolution, (obtained for example with Back Scattered
    Electron images taken with SEM), for all $k=1,\ldots,N_{{\textrm{eng}}}$ in
    general, the interaction volume exceeds a voxel in cross-sectional
    area, on the $Z$=0 surface.}}
\label{fig:3cases}
\end{figure}

\subsection{Noise in discrete image data}
\label{sec:data}
\noindent
As discussed above in Section~\ref{sec:intro}, noise in the data can
crucially influence how well-conditioned the inverse problem is.
Image data is invariably subjected to noise, though for images taken
with Scanning Electron Microscopes (SEM), the noise is small. Also,
imaging with SEM being a controlled experiment, this noise can be
reduced by the microscopist by increasing the length of time taken to
scan the sample with the electron beam that is made incident on the sample in this imaging procedure. As illustrations of the noise in real
SEM data, \ctn{scanlength} suggest that in a typical 20 second scan of
the SEM, the signal to noise is about 20:1, which can be improved to
120:1 for a longer acquisition time of 10
minutes\footnotemark. Importantly, this noise could arise due to to
variations in the beam configurations, detector irregularities,
pressure fluctuations, etc. In summary, such noise is considered to be
due to random variations in parameters that are intrinsically
continuous, motivating a Gaussian noise distribution. Thus, the noise
in the 2-D image data, created by the $i$-th beam pointing 
for $E=\epsilon_k$, (i.e. the noise in ${\tilde{I}}_i^{(k)}$) is
modelled as a random normal variable, drawn from the normal with
standard deviation $\sigma_i^{(k)}\lesssim$ 0.05$\tilde{I}_i^{(k)}$,
$k=1,\ldots,N_{{\textrm{eng}}},\:i=1,\ldots,N_{{\textrm{data}}}$. In the illustration with
real data, a scan speed of 50 s was used, which implies a noise of
less than 5$\%$.  \footnotetext{In modern SEMs, noise reduction is
  alternatively handled using pixel averaging ("Stereoscan 430
  Operator Manual", Leica Cambridge Ltd. Cambridge, England, 1994).}

\subsection{Identifying low-rank and spatially-varying components of density}
\label{sec:lowrank}
\noindent
It is known in the literature that the general, under-determined
deconvolution problem is solvable only if the unknown density is
intrinsically, ``sufficiently'' sparse \cite**{donohotanner,ma}. Here
we advance a methodology, within the frame of a designed experiment,
to learn the density - sparse or dense. In the following subsection,
we will see that priors on the sparsity that we develop here, bring in relatively more information into the models when the density structure is sparse. 

With this in mind, the density is
recognised to be made of 
a constant $\rho_0$ (the low-rank component in the limiting sense),
and the spatially varying component $\rho_1(x,y,z)$ that may be
sparse or dense in ${\mathbb R}^3$. We view the constant part of the density
as $\rho_0=\rho_0\delta(x-x_i, y-y_i, z)$ where $\delta(\cdot,\cdot,\cdot)$ is the Dirac delta function on ${\mathbb R}^3$ \ctp{chakrabortydiracdelta}, centred at the centre of the $ik$-th interaction volume, $\forall\:k=1,2,\ldots,N_{{\textrm{eng}}}$. Then in our problem, the contribution of the constant part of the
density to the projection onto the centre of the $ik$-th interaction-volume is
\begin{eqnarray}
{\cal C}(\rho_0\ast{\eta}({x,y,z})_i^{(k)})&\equiv&
\rho_0{\cal C}(\delta(x-x_i,y-y_i,z)\ast{\eta}(x,y,z))_i^{(k)}) \nonumber \\
&=& I0^{(k)},
\end{eqnarray}
a constant independent of the beam pointing location $i$, if
$\eta(x,y,z)$ is restricted to be a function of the depth coordinate
$Z$ only. As is discussed in Section~\ref{sec:correction}, this is
indeed what we adopt in the model for the kernel.

Then, $I0^{(k)}$ depends only on the known morphological details of
the interaction-volume for a given value of $E$,
$\forall\:i=1,2,\ldots,N_{{\textrm{data}}}$. Thus, $\{{\tilde
  I}_i^{(k)}\}_{i=1}^{N_{{\textrm{data}}}}= \{{\tilde I1}_i^{(k)} +
I0^{(k)}\}_{i=1}^{N_{{\textrm{data}}}}$, where ${\tilde I1}^{(k)}_i$ is the
spatially-varying component of the image data. The identification of the
constant component of the density is easily performed as due to the
constant component of the measurable. 

In our inversion exercise, it is the $\{{\tilde
  I1}_i^{(k)}\}_{i=1}^{N_{{\textrm{data}}}}$ field that is actually implemented
as data, after $I0^{(k)} := \inf\{{\tilde
  I}_i^{(k)}\}_{i=1}^{N_{{\textrm{data}}}}$ is subtracted from $\{{\tilde
  I}_i^{(k)}\}_{i=1}^{N_{{\textrm{data}}}}$, for each $k=1,\ldots,N_{{\textrm{eng}}}$.
Hereafter, when we refer to the data, the spatially-varying part of
the data will be implied; it is this part of the data that will
hereafter be referred to as
$\{{\tilde{I}}_i^{(k)}\}_{i=1}^{N_{{\textrm{data}}}}$, at each value $\epsilon_k$
of $E$, $k=1,\ldots,N_{{\textrm{eng}}}$. Its inversion will yield a spatially
varying sparse/dense density, that we will from now, be referred to as
$\rho(x,y,z)$ that in general lies in a non-convex subset of ${\mathbb
  R}_{\geq 0}$. Thus, we see that in this model, it is possible for
$\rho(x,y,z)$ to be 0. The construction of the full density, inclusive
of the low-rank and spatially-varying parts, is straightforward once
the latter is learnt.

\subsection{Basics of algorithm leading to learning of unknowns}
\label{sec:met}
\noindent
The basic scheme of the learning of the unknown functions is as
follows. First, we perform multiple projections of the convolution of
the unknown functions in the forward problem, onto the incidence point
of the $ik$-th interaction-volume, $\forall\:i,k$. The likelihood is
defined as a function of the Euclidean distance between the
spatially-averaged projection ${\cal C}(\rho\ast\eta)_i^{(k)}$, and the image datum
${\tilde I}_i^{(k)}$.  We choose to work with a Gaussian likelihood
(Equation~\ref{eqn:mainlikeli}). Since the imaging at any value of $E$
along any of the viewing vectors is done independent of all other
values of $E$ and all other viewing vectors, we assume
$\{{\tilde{I}}_i^{(k)}\}$ to be $iid$ conditionally on the values of
the model parameters.

We develop adaptive priors on sparsity on the density and present
strong to weak priors on the kernel, to cover the range of high to low
information about the kernel that may be available. In concert with
these, the likelihood leads to the posterior probability of the
unknowns, given all the image data. The posterior is sampled from
using adaptive Metropolis-Hastings within Gibbs. The unknown functions
are updated using respective proposal densities. Upon convergence, we
present the learnt parameters with 95$\%$-highest probability density
credible regions. At the same time, discussions about the
identifiability of the solutions for the 2 unknown functions and of
uniqueness of the solutions are considered in detail.

\section{Kernel function}
\label{sec:correction}
\noindent
Identifiability between the solutions for the unknown material density
and kernel is aided by the availability of measurement of the kernel at the
surface of the material sample and the model assumption that the
kernel function depends only on the depth coordinate $Z$ and is
independent of $X$ and $Y$. Thus, the correction function is
$\eta(z)$. In the discretised space of $X$, $Y$ and $Z$, the kernel is
then defined as
\begin{equation}
\eta^{(k)}:=\eta(z) \quad\mbox{for}\quad z\in[h^{(k-1)},h^{(k)}), \quad{k=1,\ldots,N_{{\textrm{eng}}}}.
\end{equation}
Then $\eta{(0)}$ is the value of the kernel function on the material surface,
i.e. at $Z=0$ and this is available from microscopy theory
\ctp{corr_epm}.

On occasions when the shape of the kernel function is known for the
material at hand, only the parameters of this shape need to be learnt
from the data. Given the measured value $\eta(0)$, i.e. the measured
$\eta^{(1)}$, the total number of kernel parameters is then 2. In this
case, the total number of parameters that we attempt learning from the
data is $N_{{\textrm{data}}}\times N_{{\textrm{eng}}} +2$, i.e. the 2 kernel shape
parameters and $N_{{\textrm{data}}}\times N_{{\textrm{eng}}}$ number of material density
parameters. We refer to such a model for the kernel as parametric. For
other materials, the shape of the kernel function may be unknown. In
that case, a distribution-free model of the $N_{{\textrm{eng}}}-1$ parameters of
the vectorised kernel function are learnt. In this case, the total
number of parameters that we learn is $N_{{\textrm{data}}}\times N_{{\textrm{eng}}}
+N_{{\textrm{eng}}}-1$.

We fall back on elicitation from the literature on electron microscopy
to obtain the priors on the kernel parameters. In microscopy practice,
the measured image data collected along an angle $\chi$ to the
vertical is given by ${\displaystyle{Q\int_{0}^{\infty} \Psi(\rho z)
    \exp[-f(\chi)\rho z]d(\rho z)}}$, where $Q\in{\mathbb R}_{> 0}$
and $f(\chi)\propto 1/\sin(\chi)$ with a proportionality constant that
is not known apriori but needs to be estimated using system-specific
Monte Carlo simulations or approximations based on curve-fitting
techniques; see \ctn{goldstein}. $\Psi(\rho z)$ is the distribution of
the variable $\rho z$ and is again not known apriori but the
suggestion in the microscopy literature has been that this can be
estimated using Monte Carlo simulations of the system or from
atomistic models. However, these simulations or model-based
calculations are material-specific and their viability in
inhomogeneous material samples is questionable. Given the lack of
modularity in the modelling of the relevant system parameters within
the conventional approach, and the existence in the microscopy
literature of multiple models that are distinguished by the underlying
approximations \cite{goldstein,corr_epm}, it is meaningful to seek to
learn the correction function.

Our construct differs from this formulation in that we construct an
infinitesimally small volume element of depth $\delta z$ inside the
material, at the point $(x,y,z)$. In the limit $\delta
z\longrightarrow$0, the density of the material inside this
infinitesimally small volume is a constant, namely
$\rho(x,y,z)$. Thus, the image datum generated from within this
infinitesimally small volume - via convolution of the material density
and the kernel - is $\rho(x,y,z)\eta(z)\delta z$. Thus, over this
infinitesimal volume, our formulation will tie in with the representation
in microscopic theory, if we set $\eta(z)\propto\Psi(\rho
z)\exp[-f(\chi)\rho z]$. It then follows that $\eta(z)$ is motivated
to have the same shape as $\Psi(\rho z)\exp(-cz)$, where $c$ is a known,
material dependent constant\footnotemark.  
When information about
this shape is invoked, the model for $\eta^{(k)}$ is referred to as
parametric; however, keeping general application contexts in
mind, a less informative prior for
the unknown kernel is also developed.
\footnotetext{$c$ depends
  upon whether the considered image is in Back Scattered Electrons
  (BSE) or X-rays. In he former case $c$ represents the BSE coefficient, often
  estimated with Monte-Carlo simulations and in the latter,
  it is the linear attenuation coefficient).}

\subsection{Parametric model for correction function}
\label{sec:secondphi}
\noindent
Information is sometimes available on the shape of the function that
represents the kernel. For example, we could follow the aforementioned 
suggestion of approximating the form of $\eta(z)$ as
$\Psi(\rho z)\exp(-cz)$, multiplied by the scale-factor $Q$, as long as $c$ is known for the material and imaging technique at hand. In fact, for
the values of beam energy $E$ that we work at, for most materials
$c\lesssim10^{-2}$, so that $\exp(-cz)\approx$ 1 ({{
    http://physics.nist.gov/PhysRefData/XrayMassCoef/tab3.html}},
Goldstein et. al 2003). Thus, we approximate the shape of $\eta(z)$ to
resemble that of $\Psi(\rho z)$, as given in microscopy
literature. This shape is akin to that of the folded normal
distribution \ctp{folded}, so that we invoke the folded normal density
function to write
\begin{equation}
\label{eqn:firstphi2}
\eta(z) \equiv \displaystyle{Q\left[
    \exp\left(-\frac{(z-\eta_0)^2}{2s^2}\right) +
    \exp\left(-\frac{(z+\eta_0)^2}{2s^2}\right)\right] }.
\end{equation}
Thus, in this model, $\eta(z)$ is deterministically known if the
parameters $Q$, $\eta_0$ and $s$ are, where, $\eta_0$ and $s$ are the
mean and dispersion that define the above folded normal
distribution. Out of these three parameters, only two are independent
since $\eta(z)$ is known on the surface, i.e. $\eta^{(1)}$
is known (see Section~\ref{sec:correction}). Thus, by setting
$\eta^{(1)}=2Q\exp[-\eta_0^2/2s^2]$, we relate $\eta_0$ deterministically
to $s$ and $Q$. In this model of the correction function, we put
folded normal priors on $Q$ and $s$. Thus,
this model of the correction function is parametric in the sense that
$\eta(z)$ is parametrised, and priors are put on its unknown
parameters. 

\subsection{Distribution-free model for correction function}
\label{sec:firstphi} 
\noindent
In this model for the correction function, we choose folded normal
priors for $\eta^{(k)}$, i.e. $\pi(\eta^{(k)})={\cal
  N}_F(m_\eta^{(k)},s)$ with $m_\eta^{(k)}=\displaystyle{Q\left[\exp\left(-\frac{(h^{(k)}-\eta_0)^2}{2s^2}\right) + \exp\left(-\frac{(h^{(k)}+\eta_0)^2}{2s^2}\right)\right]}$. The folded normal priors on $\eta^{(\cdot)}$ 
\ctp{folded} underlie the
fact that $\eta^{(k)}\geq 0,\:\forall\:k=1,\ldots,N_{{\textrm{eng}}}$ and that
there is a non-zero probability for $\eta^{(k)}$ to be zero. Thus,
folded normal and truncated normal priors for $\eta^{(k)}$ would be
relevant but gamma priors would not be. Also, microscopy theory allows
for the kernel at the surface to be known deterministically, i.e.
$\eta^{(1)}$ is known (Section~\ref{sec:correction}). Then setting
$\pi(\eta^{(1)})=$1, we relate $Q, \eta_0$ and $s$. This
relation is used to compute $s$, while uniform priors are assigned to
the hyper-parameters $\eta_0$, $Q$. We refer to this as the
``distribution-free model of $\eta(z)$''.

We illustrate the effects of both modelling strategies in simulation
studies that are presented in Section~\ref{sec:illustrations}.

The correction function is normalised by
${\widehat{\eta}}^{(1)}/\eta^{(1)}$, where ${\widehat{\eta}}^{(1)}$ is
the estimated unscaled correction function in the first $Z$-bin,
i.e. for $z\in[0,h^{(1)})$ and $\eta^{(1)}$ is the
  measured value of the correction function at the surface of the
  material, given in microscopy literature.  It is to be noted that
  the available knowledge of $\eta^{(1)}$ allows for the identifiability
  of the amplitudes of $\rho(x,y,z)$ and $\eta(z)$.

\section{Priors on sparsity of the unknown density}
\label{sec:sparse}
\noindent
In this section we present the adaptive priors on the sparsity in the
density, developed by examining the geometrical aspects of the
problem.

As we are trying to develop priors on the sparsity, we begin by
identifying the voxels in which density is zero,
i.e. $\xi_i^{(k)}=0$. This identification will be made possible by
invoking the nature of the projection operator ${\cal
  C(\cdot)}$. For example, we realise that it is possible for
the measured image datum ${\tilde I}_i^{(k)}$ collected from the
$ik$-th interaction-volume to be non-zero, even when density in the
$ik$-th voxel is zero, owing to contributions to ${\tilde I}_i^{(k)}$
from neighbouring voxels that are included within the $ik$-th
interaction-volume. Such understanding is used to identify the voxels
in which density is zero. For a voxel in which density is non-zero, we
subsequently learn the value of the density. The constraints that lead
to the identification of voxels with null density can then be
introduced into the model via the prior structure in the following ways.
\begin{enumerate}
\item We could check if ${\cal
  C}({\rho}\ast{\eta})_i^{(k)}$ = ${\cal
  C}({\rho}\ast{\eta})_{-im}^{(k)}$, (where ${\cal
  C}(\cdot)_{-im}^{(k)} :=$ projection onto the centre of the
  $ik$-th interaction-volume without including density from the
  $mk$-th voxel, i.e. without including ${\xi}_m^{(k)}$). If so, then
  ${\xi}_m^{(k)}=0$. However, this check would demand the computation
  of ${\cal C}({\rho}\ast{\eta})_i^{(k)}$ over a given
  interaction-volume, as many times as there are voxels that lie fully
  or partially inside this interaction-volume. Such computation is
  avoided owing to its being computationally
  intensive. Instead, we opt for a probabilistic suggestion for when
  density in a voxel is zero. Above, $i=1,2,\ldots,N_{{\textrm{data}}}$,
  $k=1,2,\ldots,N_{{\textrm{eng}}}$.

\item We expect that the projection of $\rho\ast\eta$ onto the centre $(x_i,y_i,0)$ of the interaction-volume achieved at $E=\epsilon_k$ will in general exceed the projection onto the same central point of a smaller interaction volume (at $E=\epsilon_{k-1}$). However, if the density of the $ik$-th voxel is zero or very low,
the contributions from the interaction-volume generated at $E=\epsilon_k$  may not exceed that from the same generated at $E=\epsilon_{k-1}$. Thus, it might be true that 
\begin{equation}
\label{eqn:constraint1}
{\cal C}({\rho}\ast{\eta})_i^{(k)} \leq {\cal C}({\rho}\ast{\eta})_i^{(k-1)} \Longrightarrow \xi_i^{(k)}=0 
  \end{equation} 
$\forall\:i,k$. This statement is true with some probability. An
alternate representation of the statement~\ref{eqn:constraint1} is
achieved as follows. We define the random variable $\tau_i^{(k)}$,
with $0< \tau_i^{(k)}\leq 1$, such that
\begin{eqnarray}
\tau_i^{(k)} &:=& \displaystyle{\frac{{\cal C}(\rho\ast\eta)_i^{(k)}}{{\cal C}(\rho\ast\eta)_i^{(k-1)}}},\:\:
{\textrm{if}}\quad {\cal C}(\rho\ast\eta)_i^{(k)} \leq {\cal C}(\rho\ast\eta)_i^{(k-1)}\quad\mbox{and}\quad {\cal C}(\rho\ast\eta)_i^{(k-1)}\neq 0 \\
\tau_i^{(k)} &:=& 1 \quad {\textrm{otherwise}} \nonumber 
\end{eqnarray}
Then the statement~\ref{eqn:constraint1} is the same as the statement:
``it might be true that $\tau_i^{(k)} \leq 1\Longrightarrow\xi_i^{(k)}
=0$ or $\xi_i^{(k)}$ is very low'', with some probability
$\nu(\tau_i^{(k)})$. In fact, as the projection from the bigger
interaction-volume is in general in excess of that from a smaller
interaction-volume, we understand that closer ${\cal
  C}(\rho\ast\eta)_i^{(k)}$ is to ${\cal C}(\rho\ast\eta)_i^{(k-1)}$,
higher is the probability that $\xi_i^{(k)}$ is close to 0. In other
words, the smaller is $\tau_i^{(k)}$, higher is the probability
$\nu(\tau_i^{(k)})$ that $\xi_i^{(k)}$ is close to 0, where
  \begin{equation}
  \label{eqn:tau}
  \nu(\tau_i^{(k)})=p^{\tau_i^{(k)}}(1-p)^{1-\tau_i^{(k)}}, 
  \end{equation}
  with the hyper-parameter $p$ controlling the non-linearity of
  response of the function $\nu(\tau_i^{(k)})$ to increase in
  $\tau_i^{(k)}$. The advantage of the chosen form of
  $\nu(\cdot)$ is that it is monotonic and its response to increasing
  value of its argument is controlled by a single parameter, namely
  $p$. We assign $p$ a a hyper prior that is uniform over the
  experimentally-determined range of [0.6, 0.99], to ensure that
  $\nu(\tau_i^{(k)})$ is flatter for lower $p$ and steeper for higher $p$, (as
  $\tau_i^{(k)}$ moves across the range $(0,1]$). The prior on the
density parameter $\xi_i^{(k)}$ is then defined as
  \begin{equation}
  \label{eqn:prior_spec}
  \pi_0(\xi_i^{(k)})= \displaystyle{{\exp\left[-\left(\xi_i^{(k)}\nu(\tau_i^{(k)})\right)^2\right]}}.
  \end{equation}
Thus, $\pi_0(\xi_i^{(k)})\in$[0,1] $\forall\:i,k$. Any normalisation
constant on the prior can be subsumed into the normalisation of the
posterior of the unknowns given the data; had we used a normalisation
of $\sqrt{4\pi/p^2}$, for all $i$, at those $k$ for which ${\cal
  C}(\rho\ast\eta)_i^{(k)} > {\cal C}(\rho\ast\eta)_i^{(k-1)}$ or
${\cal C}(\rho\ast\eta)_i^{(k-1)}=0$, the prior on $\xi_i^{(k)}$ would
have reduced to being a normal ${\cal N}(0,2/p^2)$. This is because,
for such $\xi_i^{(k)}$, $\tau_i^{(k)}=1$ so that
$\nu(\tau_i^{(k)})=p$. However, for $\tau_i^{(k)} < 1$ the prior on
$\xi_i^{(k)}$ is non-normal as the dispersion is itself dependent on
$\xi_i^{(k)}$. This prior then adapts to the sparsity of the material
density distribution. We do not aim to estimate the degree of sparsity
in the material density in our work but aim to develop a prior
$\pi_0(\xi_i^{(k)})$ $\forall\:i,k$ so that
$\{\xi_i^{(k)}\}_{k=1}^{N_{{\textrm{eng}}}}$ sampled from this prior
represents the material density distribution at the given $i$, however
sparse the vector
$(\xi_i^{(1)},\xi_i^{(2)},\ldots,\xi_i^{(N_{{\textrm{eng}}})})^T$
is. Evidently, we do not use a mixture prior but opt for no mixing as
in \ctn{greenshteinpark}. Indeed, in our prior structure, the term
$\left(\xi_i^{(k)}\nu(\tau_i^{(k)}\right)^2$ could have been replaced
by $\left\vert\xi_i^{(k)}\nu(\tau_i^{(k)}\right\vert$, (as in
parametric Laplace priors suggested by \ctn{parkcasella}, \ctn{hans},
\ctn{johnstonesilverman}), i.e. by $\xi_i^{(k)}\nu(\tau_i^{(k)})$
since $\xi_i^{(k)}$ is non-negative, but as far as sparsity tracking
is concerned--which is our focus in developing this prior here--the
prior in Equaion~\ref{eqn:prior_spec} suffices.
  \end{enumerate}
That such a prior probability density sensitively adapts to the
sparsity in the material density distribution, is brought about in the
results of 2 simulation studies shown in
Figure~\ref{fig:prior_example}. In these studies, the density
parameter values in the $ik$-th voxel are simulated from 2 simplistic
toy models that differ from each other in the degree of sparsity of
the true material density distribution: $\xi_i^{(k)}=u_1^{10}/u_2$, and
$\xi_i^{(k)}=u_3^{10}$ respectively, (where $u_1, u_2, u_3$ are
uniformly distributed random numbers in $[0,1]$), at a chosen $i$ and
energy indices $k=1,2,\ldots,10$. In the simulations we specify the
beam penetration depth $h^{(k)}\propto \epsilon_k^{1.67}$ as suggested
by \ctn{kanaya}; as any interaction-volume is hemispherical, its
radius $R0^{(k)}=h^{(k)}$.  The kernel parameters $\eta^{(k)}$ are
generated from a quadratic function of $h^{(k)}$ with noise added. In
the simulations, the material is imaged at resolution $\omega$ such
that $\pi[R0^{(10)}]^2 \leq \omega^2$, i.e. the ``1st model'' is relevant (see Section~\ref{sec:cases_new}). This allows for simplification
of the computation of ${\cal C}(\rho\ast\eta)_i^{(k)}$ according to
Equation~\ref{eqn:general}. Then at
this $i$, for $k=1,2,\ldots,10$, $\xi_i^{(k)}$ are plotted in
Figure~\ref{fig:prior_example} against $k$, as is the logarithm of the
prior $\pi_0(\xi_i^{(k)})$ computed according to
Equation~\ref{eqn:prior_spec}, with $p$ held as a random number,
uniform in [0.6,0.99]. Logarithm of the priors are also plotted as a
function of the material density parameter. We see from
the figure that the prior developed here tracks the sparsity of the
vector $(\xi_i^{(1)},\xi_i^{(2)},\ldots,\xi_i^{(N_{{\textrm{eng}}})})^T$ well.
%
\begin{figure}[!h]
     \begin{center}
  {
       \includegraphics[width=14cm]{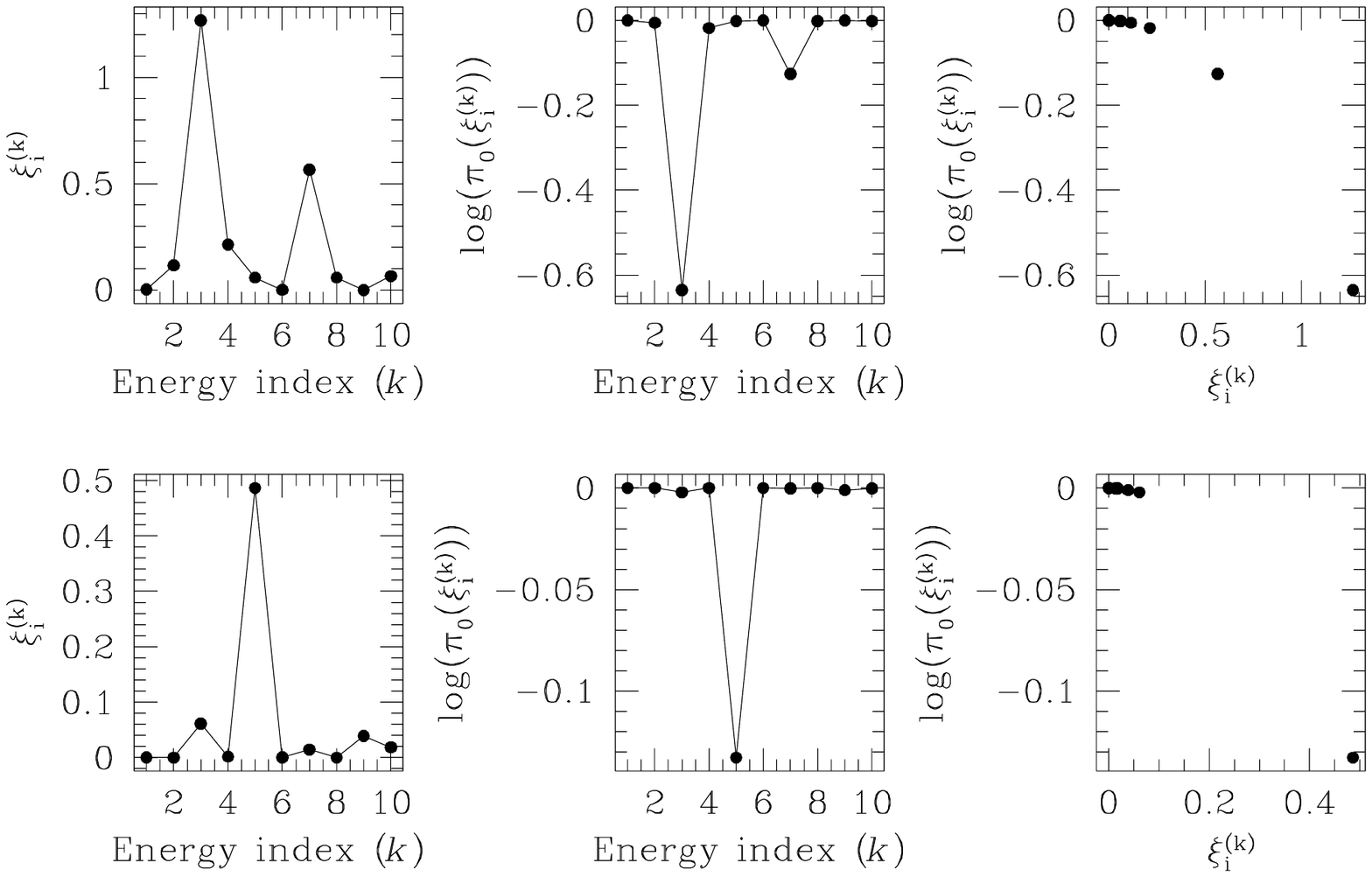}
  }
     \end{center}
\vspace*{-3in}
\caption{\small{{\it Top:} in the left panel black filled circles
    depict values of simulated material density parameters
    $\xi_i^{(k)}=u_1^{10}/u_2$, $u_1,u_2\sim{\cal U}[0,1]$, at an
    arbitrary beam position index $i$, as a function of the energy
    index $k$, for $k=1,2,\ldots,10$. Log of the prior
    $\pi_0(\xi_i^{(k)})$, as given in Equation~\ref{eqn:prior_spec},
    is shown in the middle panel as a function of $k$ for $p\sim{\cal
      U}[0.6,0.99]$. The log prior is plotted against the true values
    of $\xi_i^{(k)}$ in black filled circles in the right
    panel. 
    {\it Bottom:} As
    in the top panels, except that this simulation is of a sparser
    material density distribution with density parameters generated as
    $\xi_i^{(k)}=u_1^{10}$. 
    }}
\label{fig:prior_example}
\end{figure}

\section{Three models for image inversion}
\label{sec:model}
\noindent
In this section we develop the projection of the convolution of the unknown density and kernel onto the centre of the $ik$-th interaction volume, for each of the 3 models discussed in Section~\ref{sec:cases_new}, motivated by the 3 different classes of image data classified by their resolution.    
The difference between the size of a voxel and that of an
interaction-volume determines the difficulty in the inversion of
to the image data. 

As explained in Section~\ref{sec:cases_new}, we attempt to identify
means of dimensionality reduction, i.e. reducing the number of
integrals involved in the sequential projection of $\rho\ast\eta$,
(see Equation~\ref{eqn:general}). We do this by identifying isotropy
in the distribution of the material density within the
interaction-volume when possible, leading to elimination of the
requirement of averaging over the angular coordinate. 
  
\subsection{1st and 2nd-models - Low resolution}
\label{sec:lowresol}
\noindent
This class of image resolutions 
($\omega\sim 1\mu m$), pertain to the case when the system
is imaged by an SEM in X-rays.

\subsubsection{High-${\cal Z}$ Systems}
\label{sec:highz}
\noindent
In the 1st~model, for the high-${\cal Z}$ materials, the cross-sectional area of an
interaction-volume attained at any $E$ will fit wholly inside that of a voxel, where the cross-section is on a $Z$-constant plane.
Then at a given depth, the density inside this interaction-volume is
that inside the voxel, i.e. is a constant. Thus, dimensionality
reduction is most easily achieved in this case with the density
function isotropic inside an interaction-volume, bearing
no dependence on the angular coordinate $\theta$ of
Equation~\ref{eqn:general}. Then when we revisit Equation~\ref{eqn:general}
in its discretised form, the discrete convolution $\rho\ast\eta$ within the $k$-th $Z$-bin and the $i$-th beam pointing gives $(\rho\ast\eta)_i^{(k)} = \displaystyle{\sum_{m=1}^k \xi_i^{(m)}\eta^{(k-m)}}$, so that the projection onto the centre of the $ik$-th interaction-volume is discretised to give
in Equation~\ref{eqn:general}, gives
the following.
\begin{equation}
\label{eqn:highz_chain}
\displaystyle{{\cal C}(\rho\ast\eta)_i^{(k)}} =
\displaystyle{
\frac{1}{(R0^{(k)})^2}}
\displaystyle{
{\LARGE{\sum_{q=0}^{k}}}
\left[\frac{(R0^{(q)})^2 -(R0^{(q-1)})^2}{2} 
\left\{\sum_{t=0}^{q}\left(\left(h^{(t)} - h^{(t-1)}\right)
\displaystyle{\sum_{m=1}^t \xi^{(m)}_{i} {\eta}^{(t-m)}}\right)\right\}
                     \right]}. 
\end{equation}

\subsubsection{Low-${\cal Z}$ Systems}
\label{sec:lowz}
\noindent
In the 2nd model, for low-${\cal Z}$ materials, for any $i\in
\{1,\ldots,N_{{\textrm{data}}}\}$, and $k > k_{in}$, the cross-sectional area of the
$ik$-th interaction-volume on the $Z$=0 plane will spill out of the
$i1$-th voxel into the neighbouring voxels. Then only for at
$k=1,2,\ldots,k_{in}$, isotropy within the $ik$-th interaction volume
holds good. In general, the projection ${\cal
  C}(\rho\ast\eta)_i^{(k)}$ onto the centre of the $ik$-th
interaction-volume includes contributions from all those voxels that
lie wholly as well as partly inside this interaction volume. This
projection is then computed by first learning the weighted average of
the contributions from all the relevant voxels and then distributing
this learnt weighted average over the identified voxels in the
proportion of the weights which in turn are in proportion of the voxel
volume that lies within the $ik$-th interaction-volume. Thus, for the
$i' k'$-th voxel that is a neighbour of the $ik'$-th voxel that lies
wholly inside the $ik$-th interaction-volume, let this proportion
be $w^{(k')}_{i'\vert i}$, where $k'\leq k$.

In general, at a given $Z$, any bulk voxel has 8
neighbouring voxels and when the voxel lies at the corner or edge of
the sample, number of nearest
neighbours is less than 8. Then, at a given $Z$, there will be contribution from at most 9 voxels
towards ${\cal
  C}(\rho\ast\eta)_i^{(k)}$. At $Z=z\in[h^{(k'-1)},h^{(k')}]$, for any $i$, let the maximum number of nearest neighbours be $i_{max}\vert i,k$ so that $i_{max}\vert i,k \leq 9$. The notation for this number bears its dependence on both $i$ and $k$. 
We define
${\bar\xi}_i^{(k')}$ as the weighted average of the densities in the
$ik'$-th voxel and its nearest neighbours that are fully or partially
included within the $ik$-th interaction-volume. Here $k' \leq k$,
$k=1,\ldots,N_{{\textrm{eng}}}$, $i=1,\ldots,N_{{\textrm{data}}}$. Thus, \\
\begin{equation}
\label{eqn:chi_lolo}
{\bar\xi}_i^{(k')}:=
  \displaystyle{\sum_{i'=1}^{i_{max}\vert i,k'}\xi_{i'\vert i}^{(k')}
    w^{(k')}_{i'\vert i}},
\end{equation}
where the $i'$-th neighbour of the $ik'$-th voxel at the same depth,
harbours the density $\xi_{i'\vert i}^{(k')}$ and there is a maximum
of $i_{max}\vert i,k'$ such neighbours. The effect of this averaging
over the nearest neighbours at this depth, is equivalent to averaging
over the angular coordinate $\theta$ and results in the angular
averaged density ${\bar\xi}_i^{(k')}$ at this $Z$, which by
definition, is isotropic, i.e. independent of the angular
coordinate. Then the projection onto the centre of the $ik$-th
interaction-volume, for $k > k_{in}$, is computed as in
Equation~\ref{eqn:highz_chain} with the material density term
$\xi_i^{(\cdot)}$ on the RHS of this equation replaced by the
isotropic angular averaged density ${\bar\xi}_i^{(\cdot)}$. However,
for $k \leq k_{in}$, the projection is computed as in
Equation~\ref{eqn:highz_chain}.


\subsection{3$^{rd}$ model - Fine resolution}
\label{sec:fine}
\noindent
In certain imaging techniques, such as imaging in Back Scattered
Electrons by an SEM or FESEM, the resolution $\omega \ll
R0_i^{(N_{{\textrm{eng}}})}$, $\forall\: i = 1,\ldots, N_{{\textrm{data}}}$. In this case,
at a given $Z$, material density from multiple voxels contribute to
${\cal C}(\rho\ast\eta)_i^{(k)}$ and the three integrals involved in
the computation of this projection, as mentioned in
Equation~\ref{eqn:general}, cannot be avoided. Knowing the shape of
the interaction-volume, it is possible to identify voxels that
live partly or wholly inside the $ik$-th interaction-volume as well as
compute the fractional or full volume of each such voxel inside the $ik$-th
interaction volume.

For this model, the projection equation is written in terms of the
coordinates $(x,y,z)$ of a point instead of the polar coordinate
representation of this point, where the point in question lies inside
the $ik$-th interaction-volume that is centred at $(x_i,y_i,0)$. Then
inside the $ik$-th interaction-volume, at a given $x$ and $y$,
$z\in\left[0,\sqrt{\left(R0^{(k)}\right)^2 - (x-x_i)^2 -(y-y_i)^2}\right]$. For \\
\begin{equation}
x-x_i\in[(u-1)\omega, u\omega] \quad
  u=\displaystyle{-(int)\left(\frac{R0^{(k)}}{\omega}\right)+1,
    -(int)\left(\frac{R0^{(k)}}{\omega}\right)+2,\ldots,(int)\left(\frac{R0^{(k)}}{\omega}\right)},\nonumber
\end{equation}
  the index $p_u(k)$ of the $Y$-bin of voxels lying fully inside the
  $ik$-th interaction volume, with respect to the centre of this interaction-volume, are
\begin{equation}
  p_u(k)=-q_u(k),-q_u(k)+1,\ldots,0,1,2,\ldots,q_u(k)-1,q_u(k), \nonumber
\end{equation} where
\begin{equation}  
q_u(k):=\displaystyle{(int)\left(\frac{\sqrt{(R0^{(k)})^2 - u^2\omega^2}}{\omega}\right)}.\nonumber
\end{equation}
Then using the definition of the beam-pointing index in terms of the $X$-bin and $Y$-bin indices of voxels (see Equation~\ref{eqn:modulo}), we get the beam-pointing index $\varrho_u(i,k)$ of voxels lying wholly inside the $ik$-th interaction-volume, for a given $u$ is 
\begin{equation}
\varrho_u(i,k) = i-q_u(k)\sqrt{N_{{\textrm{data}}}}+u,\:
i-\left(q_u(k)-1\right)\sqrt{N_{{\textrm{data}}}}+u,\:
\ldots, \:
i-\left(q_u(k)-2q_u(k)\right)\sqrt{N_{{\textrm{data}}}}+u, \nonumber
\end{equation} 
i.e. for a given $u$, $\varrho_u(i,k)=i+p_u(k)\sqrt{N_{{\textrm{data}}}}+u$.

  The depth coordinate of voxels with beam-pointing index $\varrho_u(i,k)$
  lying inside the $ik$-th interaction-volume are
  $z\in\displaystyle{\left[0, {\sqrt{(R0^{(k)})^2 - (p_u(k))^2
          \omega^2 - u^2\omega^2}}\right]}$ so that the energy index of voxels lying fully
  inside at $Y$-bin $p_u(k)$ and $x-x_i\in[(u-1)\omega,u\omega)$ are $\in[1,t_{max}(u)]$ where $t_{max}(u)\in{\mathbb Z}_{> 0}$ such that
$t_{max}(u)=\max\{1,2,\ldots,N_{{\textrm{eng}}}\}$ that satisfies
  \begin{equation}
  \displaystyle{h^{(t_{max}(u))}}\leq {\sqrt{(R0^{(k)})^2 - (p_u(k))^2
          \omega^2 - u^2\omega^2}}.\nonumber
\end{equation} 
At this $Y$-bin index $p_u(k)$, there will also exist a voxel lying
partly inside the $ik$-th interaction-volume, at the
$(t_{max}(u)+1)$-th $Z$-bin, between depths $h^{t_{max}(u)}$
and ${\sqrt{(R0^{(k)})^2 - (p_u(k))^2 \omega^2 - u^2\omega^2}}$. In addition, the
projection ${\cal C}(\rho\star\eta)_i^{(k)}$ will include
contributions from voxels at the edge of this interaction-volume,
lying partly inside it; the beam-pointing indices of such voxels will
be $i-\left(q_u(k)+1\right)\sqrt{N_{{\textrm{data}}}}+u$ and $i+\left(q_u(k)+1\right)\sqrt{N_{{\textrm{data}}}}+u$ for
$x-x_i\in[(u-1)\omega,u\omega]$ with $u$ and $q(u)$ defined as
above. Lastly, parts of voxels at beam-pointing indices
$i-\displaystyle{(int)\left(\frac{R0^{(k)}}{\omega}\right)}-1$ and
$i+\displaystyle{(int)\left(\frac{R0^{(k)}}{\omega}\right)}+1$ will also
be contained inside the $ik$-th interaction-volume. These voxels at the edges extend into the 1st $Z$-bin. We can compute the fraction $r_{a}^{(b)}(i,k)$ of the volume of the $ab$-th voxel contained partly within the $ik$-th interaction-volume by tracking the geometry of the system. Then using the discretised version of Equation~\ref{eqn:general}, we write,\\
\hspace*{.8cm}
$\displaystyle{\omega^{-2}(R0^{(k)})^2 {\cal C}(\rho\ast\eta)_i^{(k)}} =$
\begin{eqnarray}
\label{eqn:lowres2} 
&& \displaystyle{\sum_{u=-(int)(R0^{(k)}/\omega)}^{(int)(R0^{(k)}/\omega)}\:\:
\sum_{p_u(k)=-q_u(k)}^{q_u(k)} 
\sum_{t=1}^{t_{max}(u)}
\left[
\left(h^{(t)}-h^{(t-1)}\right)
                           \sum_{m=1}^t {\xi}_{\varrho_u(i,k)}^{(m)} {\eta}^{(t-m)}\right]} + \nonumber \\
&&
              \displaystyle{
\sum_{u=-(int)(R0^{(k)}/\omega)}^{(int)(R0^{(k)}/\omega)}\:\:
\sum_{p_u(k)=-q_u(k)}^{q_u(k)} 
\left[
\left(\sqrt{(R0^{(k)})^2-((q_u(k))^2+u^2)\omega^2} - h^{(t_{max}(u))}\right)
     \sum_{m=1}^{t_{max}(u)+1}  \chi_{\varrho_u}^{(m)}(i,k){\eta}^{(t_{max}(u)+1-m)}\right]} + \nonumber \\
&&              \displaystyle{
\sum_{\ell(i,k)}
\left[
\left(h^{(1)}\right)
      {r}_{\ell}^{(1)}(i,k) {\xi}_{\ell(i,k)}^{(1)} {\eta}^{(0)}\right]} 
\end{eqnarray}
where\: $\ell(i,k)=i-\displaystyle{(int)\left(\frac{R0^{(k)}}{\omega}\right)}-1, i+\displaystyle{(int)\left(\frac{R0^{(k)}}{\omega}\right)}+1, i-\left(q_u(k)+1\right)\sqrt{N_{{\textrm{data}}}}+u, i+\left(q_u(k)+1\right)\sqrt{N_{{\textrm{data}}}}+u$, for $u=\displaystyle{-(int)\left(\frac{R0^{(k)}}{\omega}\right)}+1,\displaystyle{-(int)\left(\frac{R0^{(k)}}{\omega}\right)}+2,\ldots,\displaystyle{(int)\left(\frac{R0^{(k)}}{\omega}\right)}$,\\
$\chi_{\varrho_u}^{(m)}(i,k):=r_{\varrho_u}^{(m)}(i,k){\xi}_{\varrho_u(i,k)}^{(m)}$, \\
and $\eta^{(1)}$ is the measured value of the kernel on the system surface (see Section~\ref{sec:correction}).


\section{Inference}
\label{sec:inference}
\noindent
In this work, we learn the unknown material density and kernel
parameters using the mismatch between the data $\{{\tilde
  I}_i^{(k)}\}^{k=N_{{\textrm{eng}}};\:i=N_{{\textrm{data}}}}_{k=1;\:i=1}$ and
$\{{\cal
  C}(\rho\ast\eta)_i^{(k)}\}^{k=N_{{\textrm{eng}}};\:i=N_{{\textrm{data}}}}_{k=1;\:i=1}$, in
terms of which, the likelihood is defined.  The material density and
kernel are convolved, and this convolution is sequentially projected
onto the centre of the the $ik$-th interaction volume, in the model
(out of the 3 models, depending on the resolution of the image data
at hand). Thus, if for the given material, the available data is such
that $\omega^2 \geq \pi[R0^{(N_{{\textrm{eng}}})}]^2$, then we use
Equation~\ref{eqn:highz_chain} to implement ${\cal
  C}(\rho\ast\eta)_i^{(k)}$.  If $\omega^2 \geq \pi[R0^{(k_{in})}]^2$
but $\omega^2 < \pi[R0^{(k_{in}+1)}]^2$ , Equation~\ref{eqn:highz_chain}
is relevant while for $\omega^2 \leq \pi[R0^{(k)}]^2$
the nearest neighbour averaging is invoked; see Section~\ref{sec:lowz}.

We choose to work with a Gaussian likelihood:
\begin{eqnarray}
\label{eqn:mainlikeli}
\\
&&{\cal L}\left(\xi_1^{(1)},\ldots,\xi_1^{(N_{{\textrm{eng}}})},\ldots,\xi_{N_{{\textrm{data}}}}^{(1)},\ldots,\xi_{N_{{\textrm{data}}}}^{(N_{{\textrm{eng}}})},\eta^{(1)},\ldots,\eta^{(N_{{\textrm{eng}}})}\vert {\tilde I}_1^{(1)}, {\tilde I}_1^{(2)},\ldots,{\tilde I}_1^{(N_{{\textrm{eng}}})},{\tilde I}_2^{(1)},\ldots,{\tilde I}_2^{(N_{{\textrm{eng}}})},\ldots,{\tilde I}_{N_{{\textrm{data}}}}^{(N_{{\textrm{eng}}})}\right) = \nonumber \\
&&\displaystyle{\prod_{k=1}^{N_{{\textrm{eng}}}}\prod_{i=1}^{N_{{\textrm{data}}}}\frac{1}{\sqrt{2\pi}\sigma_i^{(k)}}\exp\left[-\frac{\left({\cal
  C}(\rho\ast\eta)_i^{(k)} - {\tilde I}_i^{(k)}\right)^2}{2\left(\sigma_i^{(k)}\right)^2}\right]}, 
\end{eqnarray}
where the noise in the image datum ${\tilde{I}}_i^{(k)}$ is
$\sigma_i^{(k)}$; it is discussed in Section~\ref{sec:data}.

Towards the learning of the unknown functions, the joint posterior
probability density of the unknown parameters, given the image data,
is defined using Bayes rule as
\begin{eqnarray}
\label{eqn:posterior}
\\
&&\displaystyle{\pi\left(\xi_1^{(1)},\ldots,\xi_1^{(N_{{\textrm{eng}}})},\ldots,\xi_{N_{{\textrm{data}}}}^{(1)},\ldots,\xi_{N_{{\textrm{data}}}}^{(N_{{\textrm{eng}}})},\eta^{(1)},\ldots,\eta^{(N_{{\textrm{eng}}})}\vert {\tilde I}_1^{(1)},\ldots,{\tilde I}_{N_{{\textrm{data}}}}^{(N_{{\textrm{eng}}})}\right)} \propto \nonumber \\
&&\displaystyle{{\cal L}\left(\xi_1^{(1)},\ldots,\xi_1^{(N_{{\textrm{eng}}})},\ldots,\xi_{N_{{\textrm{data}}}}^{(1)},\ldots,\xi_{N_{{\textrm{data}}}}^{(N_{{\textrm{eng}}})},\eta^{(1)},\ldots,\eta^{(N_{{\textrm{eng}}})}\vert {\tilde I}_1^{(1)},\ldots,{\tilde I}_1^{(N_{{\textrm{eng}}})},\ldots,{\tilde I}_2^{(N_{{\textrm{eng}}})},\ldots,{\tilde I}_{N_{{\textrm{data}}}}^{(N_{{\textrm{eng}}})}\right)\times} \nonumber \\
&& \displaystyle{\pi_0\left(\xi_1^{(1)},\xi_1^{(2)},\ldots,\xi_{N_{{\textrm{data}}}}^{(1)},\ldots,\xi_{N_{{\textrm{data}}}}^{(1)},\ldots,\xi_{N_{{\textrm{data}}}}^{(N_{{\textrm{eng}}})}\right)
\nu_0\left(\eta^{(1)},\ldots,\eta^{(N_{{\textrm{eng}}})}\right)} \nonumber 
\end{eqnarray}
where $\pi_0(\xi_1^{(1)},\ldots,\xi_{N_{{\textrm{data}}}}^{(N_{{\textrm{eng}}})})$ is the
adaptive prior probability on the sparsity of the density function, as
discussed in Section~\ref{sec:sparse}. Also,
$\nu_0(\eta^{(1)},\ldots,\eta^{(N_{{\textrm{eng}}})})$ is the prior on the
kernel, discussed in Section~\ref{sec:correction}.

Once the posterior probability density of the material density function and
kernel, given the image data is defined, we use the adaptive
Metropolis within Gibbs \ctp{haario} to generate
posterior samples.

At the $n$-th iteration, $n=1,\ldots,N_{max}$, $\xi_i^{(k)}$ is
proposed from a folded normal density\footnotemark. This choice of the
proposal density is motivated by a non-zero probability for
$\xi_i^{(k)}$ to be zero. The latter constraint rules out a gamma or
beta density that $\xi_i^{(k)}$ is proposed from but truncated and
folded normal densities are acceptable
$k=1\ldots,N_{{\textrm{eng}}},\:i=1,\ldots,N_{{\textrm{data}}}$. Of these we choose the
easily computable folded normal proposal density \ctp{folded}. The
proposed density in the $n$-th iteration, in the $ik$-th voxel is
\begin{equation}
\label{eqn:qdefn}
{\tilde\xi}_i^{(k)}\vert_n \sim {\cal N}_F(\mu_i^{(k)}\vert_n, \varsigma_i^{(k)}\vert_n) 
\end{equation}
while the current density in this voxel at the $n$-th iteration is
defined as ${\xi}_i^{(k)}\vert_{n}$. \footnotetext{The distribution
  ${\cal N}_F(a, b)$ is the folded normal distribution with mean
  $a\in{\mathbb R},\: a > 0$ and standard deviation $b\in{\mathbb
    R},\:b>0$ \ctp{folded}}. We choose the mean and variance of this
proposal density to be
\begin{eqnarray}
\\
\mu_i^{(k)}\vert_n &=& \xi_i^{(k)}\vert_{n-1},\quad \forall\:n=1,\ldots,N_{max} \nonumber \\
\left(\varsigma_i^{(k)}\vert_n\right)^2  &=&  
\begin{cases}
\displaystyle{\frac{\sum_{p=n_0}^{n-1}\left(\xi_i^{(k)}\vert_p\right)^2}{n-n_0} -
\left[\frac{\sum_{p=n_0}^{n-1}\left(\xi_i^{(k)}\vert_p\right)}{n-n_0}\right]^2}\quad \textrm{if}\quad n\geq n_0 \nonumber \\
T\xi_i^{(k)}\vert_0 \quad{\textrm{if}}\quad n<n_0  \nonumber 
\end{cases}
\end{eqnarray}
The random variable $T$ is considered to be uniformly distributed,
i.e. $T\sim U(0,1]$.  Thus, for $n \geq n_0$,
  the proposal density is adaptive, \cite**{haario}. We choose
  $n_0=10^4$ and $N_{max}$ is of the order of 8$\times 10^{5}$.

We choose $\xi_i^{(k)}\vert_0$ by assigning
  constant density to the voxels that constitute the $ik$-th
  interaction-volume,
  $k=1\ldots,N_{{\textrm{eng}}},\:i=1,\ldots,N_{{\textrm{data}}}$. 

When a distribution-free model for the kernel is used, in
the $n$-th iteration, $\eta^{(k)}$ is proposed from exponential proposal
density with a constant rate parameter $s_1$. 
When the parametric model for the kernel is used,
$\eta(z)$ is calculated as given in Equation~\ref{eqn:firstphi2},
conditional on the values of 2 the parameters $Q$ and $\eta_0$. The
proposed parameters at the $n$-th iteration are ${\tilde Q}_n$ and
$({\tilde\eta_0})_n$. ${\tilde Q}_n$ and $({\tilde\eta_0})_n$ are each
proposed from independent exponential proposal densities with constant
rate parameters.

Inference is
performed by sampling from the high dimensional posterior
(Equation~\ref{eqn:posterior}) using Metropolis-within-Gibbs block
update, \cite**{gilks96,chibgreenberg95}. Let the state at the $n$
iteration be
\begin{equation}
\varepsilon_n=(\xi_1^{(1)}\vert_n,\ldots,\xi_1^{(N_{{\textrm{eng}}})}\vert_n,\ldots,\xi_2^{(N_{{\textrm{eng}}})}\vert_n, \ldots, \xi_{N_{{\textrm{data}}}}^{(N_{{\textrm{eng}}})}\vert_n,\eta^{(1)}\vert_n,\ldots,\eta^{(N_{{\textrm{eng}}})}\vert_n)^T.
\end{equation}
For the implementation of the block Metropolis-Hastings, we partition the
state vector $\varepsilon_n$ as: 
\begin{equation}
\varepsilon_n^T = ((\varepsilon_n^{(\xi)})^T, (\varepsilon_n^{(\eta)})^T),\nonumber
\end{equation} 
where
\begin{eqnarray}
\varepsilon_n^{(\xi)} &=& (\xi_1^{(1)}\vert_n,\ldots,\xi_1^{(N_{{\textrm{eng}}})}\vert_n,\ldots,\xi_2^{(N_{{\textrm{eng}}})}\vert_n, \ldots, \xi_{N_{{\textrm{data}}}}^{(N_{{\textrm{eng}}})}\vert_n)^T,\nonumber \\
\varepsilon_n^{(\eta)} &=& (\eta^{(1)}\vert_n,\ldots,\eta^{(N_{{\textrm{eng}}})}\vert_n)^T.
\end{eqnarray}
Here $n=1,\ldots,N_{burn_{in}},\ldots,N_{max}$. We typically use
$N_{max} >$ 8$\times 10^5$ and $N_{burn_{in}}$=1$\times 10^5$. Then,
the state $\varepsilon_{n+1}$ is given by the successive updating of
the two blocks: ${\varepsilon}_{n+1}^{(\xi)}$ {{and}}
${\varepsilon}_{n+1}^{(\eta)}$.


\section{Posterior probability measure in small noise limit}
\label{sec:posterior}
\noindent
Here we check for uniqueness in the learnt $\rho\ast\eta$. 
With the aim of investigating the posterior probability density in the
small noise limit, we recall the chosen priors for the material density
(Section~\ref{sec:sparse}), kernel parameters
(Section~\ref{sec:correction}), and the likelihood 
function (Equation~\ref{eqn:mainlikeli}).
\begin{theorem}
\label{theorem:unique}
In the limit of small noise, $\sigma_i^{(k)}\longrightarrow 0$, the
joint posterior probability of the density and
kernel, given the image data, for all beam-pointing vectors
($i=1,\ldots,N_{{\textrm{data}}}$) and all $\epsilon_k,\:k=1,\ldots,N_{{\textrm{eng}}}$,
reduces to a product of $N_{{\textrm{data}}}\times N_{{\textrm{eng}}}$ Dirac measures, with
the $ik$-th measure centred at the solution to the equation ${\tilde
  I}_i^{(k)} = {\cal C}(\rho\ast\eta)_i^{(k)}$,
\label{theorem:main}
\end{theorem}

\begin{proof}
Logarithm of the posterior probability of the discretised
distribution-free model is
\begin{eqnarray}
\label{eqn:posteriorjalaton}
&&\log\pi\left(\xi_1^{(1)},\ldots,\xi_1^{(N_{{\textrm{eng}}})},\ldots,\xi_{N_{{\textrm{data}}}}^{(1)},\ldots,\xi_{N_{{\textrm{data}}}}^{(N_{{\textrm{eng}}})},\eta^{(1)},\ldots,\eta^{(N_{{\textrm{eng}}})}\vert {\tilde I}_1^{(1)},\ldots,{\tilde I}_{N_{{\textrm{data}}}}^{(N_{{\textrm{eng}}})}\right)
=  \nonumber \\
&&\displaystyle{
\sum_{i=1}^{N_{{\textrm{data}}}}\sum_{k=1}^{N_{{\textrm{eng}}}}
\left[-\log{\sigma_i^{(k)}}
- \left(\frac{({\tilde I}_i^{(k)} -{\cal C}(\rho\ast\eta)_i^{(k)})^2}{2(\sigma_i^{(k)})^2} 
 \right)\right]}
- \displaystyle{\sum_{k=1}^{N_{{\textrm{eng}}}}\left[\frac{(\eta^{(k)}+\eta_0^{(k)})^2}{2N(s^{(k)})^2}\right ]} 
- \sum_{i=1}^{N_{{\textrm{data}}}}\sum_{k=1}^{N_{{\textrm{eng}}}}\left[
{\left(\xi_i^{(k)}\nu(\tau_i^{(k)})\right)^2}\right] + A, \nonumber \\ 
&& 
\end{eqnarray}
where $A\in{\mathbb R}$ is a finite constant. Thus, 
\begin{eqnarray}
\label{noiselimit}
\displaystyle{
\lim_{\sigma_i^{(k)}\longrightarrow 0} \pi(\xi_1^{(1)},\ldots,\xi_{N_{{\textrm{data}}}}^{(N_{{\textrm{eng}}})},\eta^{(1)},\ldots,\eta^{(m)}
\vert{\tilde I}_1^{(1)},\ldots,{\tilde I}_{N_{{\textrm{data}}}}^{(N_{{\textrm{eng}}})})} \propto && \nonumber \\
\displaystyle{\lim_{\sigma_i^{(k)}\to 0}\left[\prod\limits_{i=1}^{N_{{\textrm{data}}}}
                                              \prod\limits_{k=1}^{N_{{\textrm{eng}}}} 
\frac{1}{\sigma_i^{(k)}}\exp\left(-\frac{({\tilde I}_i^{(k)} -{\cal C}(\rho\ast\eta)_i^{(k)})^2}{2(\sigma_i^{(k)})^2} 
 \right)\right]}. && 
\end{eqnarray}
The right hand side of this equation is the product of Dirac delta functions
centred at ${\tilde I}_i^{(k)} = {\cal C}(\rho\ast\eta)_i^{(k)}$,
for $i=1,\ldots,N_{{\textrm{data}}},\:k=1,\ldots,N_{{\textrm{eng}}}$. Thus, the
posterior probability density reduces to a product of Dirac measures
for each $i,k$, with each measure centred on the solution of the equation
${\tilde  I}_i^{(k)} = {\cal C}(\rho\ast\eta)_i^{(k)}$. 
\end{proof}

On the basis of Theorem~\ref{theorem:unique}, we arrive at the following 
important results.
\begin{enumerate}
\item[(i)] $(\rho\ast\eta)_i^{(k)}$ is the least-squares solution to
  ${\tilde I}_i^{(k)} = {\cal C}(\rho\ast\eta)_i^{(k)}$.
\item[(ii)] {The three contractive projections that act successively to project
  $\rho\ast\eta$ onto the space of functions defined over the point of
  incidence of any interaction-volume, are commutable. Thus, for
  example, the result is invariant to whether the projection onto the
  $Z$=0 plane happens first and then the projection onto the $Y$-axis
  happens or whether first the projection onto the $Y$=0 plane is
  performed, followed by the projection onto the $Z$=0 axis. Each of
  these projections is also orthogonal and is therefore represented by
  a projection matrix ${\bf P}_i$, $i=1,2,3$, that is idempotent
  (${\bf P}_i={\bf P}_i^2$) and symmetric. Then, the composition of
  these projections, i.e. the ${\cal C}$ operator, which in the matrix
  representation is ${\bf C}$, is a product of three idempotent and
  symmetric matrices that commute with each other. This implies that
  ${\bf C}$ is idempotent and symmetric, implying that ${\cal C}$ is
  an orthogonal projection.}
\item[(iii)] In the small noise limit, $(\rho\ast\eta)_i^{(k)}$ is the
  unique solution to the least-squares problem ${\tilde
    I}_i^{(k)}={\cal C}(\rho\ast\eta)_i^{(k)}$. The justification
  behind this results is as follows. In the matrix representation,
  ${\bf C}$ is a square matrix of dimensionality ${N_{{\textrm{eng}}}\times
    N_{{\textrm{eng}}}}$, where the image data is a $N_{{\textrm{eng}}}\times N_{{\textrm{data}}}$
  dimensional matrix and $\rho\ast\eta$ is also represented by a
  $N_{{\textrm{eng}}}\times N_{{\textrm{data}}}$ dimensional matrix. ${\cal C}$ being a
  composition of three commutable orthogonal projections, the matrix
  ${\bf C}^{(N_{{\textrm{eng}}}\times N_{{\textrm{eng}}})}$ is idempotent and
  symmetric. Then the Moore-Penrose pseudo-inverse of ${\bf C}$ exists
  as ${\bf C}^{+}$, and is unique for given ${\bf C}$. Then
  $(\rho\ast\eta)_i^{(k)}$ is the unique solution to the least squares
  problem ${\tilde{I}}_i^{(k)} = {\bf C}(\rho\ast\eta)_i^{(k)}$
  achieved in the low noise limit, such that $(\rho\ast\eta)_i^{(k)} =
  {\bf C}^{+} {\tilde{I}}_i^{(k)}$.
\item[(iv)] {The establishment of the uniqueness of the learnt
  $\rho\ast\eta$ in the small noise limit is due to our
  designing of the imaging experiment to ensure that the
  dimensionality of the image space coincides with that of the space
  of the unknown parameters.}
\item[(v)] In the presence of noise in the data, the learnt
  $\rho\ast\eta$ is no longer unique. The variation in this learnt
  function is then given by the condition number $\kappa({\bf C})$ of
  ${\bf C}$. Now, ${\kappa}({\bf C})=\parallel {\bf C}\parallel
  \parallel {\bf C}^{+}\parallel$, where $\parallel\cdot\parallel$
  here refers to the 2-norm. For orthogonal projection matrices, norm
  is 1, i.e. $\kappa({\bf C})$=1. Therefore, fractional deviation of
  uniqueness in learnt value of $\rho\ast\eta$ is the same as the
  noise in the image data, which is at most 5$\%$ (see
  Section~\ref{sec:data}).
\end{enumerate}

\subsection{Quantification of deviation from uniqueness of learnt functions}
\label{sec:quant}
\noindent
While $\rho\ast\eta$ is learnt uniquely in the small noise limit, the
learning of $\rho(x,y,z)$ and $\eta(z)$ from this unique convolution is 
classically, an
ill-posed problem since then we know $N_{{\textrm{data}}}\times N_{{\textrm{eng}}}$
parameters but $(N_{{\textrm{data}}}\times N_{{\textrm{eng}}}) + N_{{\textrm{eng}}}$ parameters are
unknown in the implementation of the distribution-free model of $\eta(z)$;
again $(N_{{\textrm{data}}}\times N_{{\textrm{eng}}}) + 2$ are unknown in case the
parametric model of $\eta(z)$ is used. Then the ratios of the
known to unknown parameters in these two cases are as high as $r_{nonpram}
\approx 0.990$ and $r_{param} \approx 0.998$ respectively, for typical values of $N_{{\textrm{data}}}$=100, $N_{{\textrm{eng}}}$=10. 

In a Bayesian framework, the problem is addressed via the priors on
the unknown parameters presented above. The prior on the material
density parameters is adaptive to the sparsity of the density in its
native space. Thus, this prior should expectedly perform equally well
when the density function is intrinsically dense as well as for
density functions marked by high degrees of sparsity. This is evident
in the quality of comparison of the true densities and the densities
learnt from simulated image data sets obtained by sampling from two
density functions that are equivalent in all respect except for their
inherent sparsity, as shown in Figure~\ref{fig:colour}. At the same
time, the parametric model for the kernel is more informed than the
distribution-free model. We find corroboration of improved inference
on the unknowns when the parametric, rather than distribution-free
model for the kernel parameters is implemented, in learning with
simulated data, (see Figures~\ref{fig:1d} and \ref{fig:sem}).

\section{Inversion of simulated microscopy image data}
\label{sec:illustrations}
We describe the application of the inverse methodology described above
to simulated images, as a test of the method, i.e. compare the true
material density and true kernel parameters (or correction
function) with the estimated density and kernel respectively. Here,
the ``true'' density and kernel are the chosen density and kernel
functions using which the simulated image data are constructed.

The discussed examples include inversion of simulated image data 
\begin{itemize}
\item of an example of a high-${\cal Z}$ material, namely Copper-Tungsten alloys, the true density function of which is
\begin{enumerate}
\item[--]dense (Sample~I-CuW), 
\item[--]sparse (Sample~II-CuW), such that the density structure is characterised by isolated modes with sharp edges. 
\end{enumerate}
\item of a low-${\cal Z}$ material, namely a Ni-Al alloy which is 
\begin{enumerate}
\item[--]dense (Sample~I-NiAl), 
\item[--]sparse (Sample~II-NiAl).
\end{enumerate}
\end{itemize}

The (simulated) images are produced, in some cases at 18 beam energy
values and in other cases at 10 different values of $E$, corresponding
to $\epsilon_k=n+k$, in real physical units of energy, namely
kiloVolts (or kV), $k=1,\ldots,N_{{\textrm{eng}}}$, $n\in{\mathbb R}_{\geq
  0}$. We work with $N_{{\textrm{data}}}$=225, i.e. there are 15 pixels along the
$X$ and 15 along the $Y$-axis corresponding to 15 beam pointings along
each of these axes. Then at any $k$, the image data in 225 pixels, are
arranged in a 15$\times$15 square array. The image data in each pixel
is computed from a chosen density and chosen correction function by
allowing the sequential projection operator ${\cal C}$ to operate upon
$\rho\ast\eta$ where the material density parameters are chosen as
\begin{eqnarray}
\xi_i^{(k)} = \displaystyle{\Upsilon\frac{A_i}{\displaystyle{\left[\epsilon^2 +
                                              \frac{x_i^2}{B_i^2} +
                                              \frac{y_i^2}{B_i^2} +
                                              \frac{(h^{(k)})^2}{B_i^2 (1 - Q_i^2)}
                                         \right]}}
                             }&&\quad{\textrm{where}} 
\label{eqn:trueden}
\end{eqnarray}
$\Upsilon = ({\textrm{int}})(N_{{\textrm{eng}}} U_1),\quad U_1\sim{\cal U}[0,1]\quad{\textrm{for the sparse Sample~II-CuW and Sample~II-NiAl}}$,\\ 
$\Upsilon = 1 \quad{\textrm{for the dense Sample~I-CuW and Sample I-NiAl}}$,\\
$A_i = U_2,\quad U_2\sim{\cal U}[0,1]$,\\ 
$B_i = U_3, \quad U_3\sim{\cal U}[0,n\omega]\quad n\in{\mathbb{Z}}^{+}$,\\
$Q_i = U_4,\quad U_4\sim{\cal U}[0,1], \forall\:i=1,\ldots,N_{{\textrm{data}}}$.\\
Here ${\cal U}[0,\cdot]$ is the uniform distribution over the range
$[0,\cdot]$.

The true correction function is chosen to emulate a folded normal
distribution with a mean of $\gamma$ and dispersion $d_s$.
\begin{equation}
\eta^{(k)} = \displaystyle{\exp\left[\frac{(h^{(k)} - \gamma)^2}{2d_s^2}\right] +
                        \exp\left[\frac{(h^{(k)} + \gamma)^2}{2d_s^2}\right] }
\end{equation}

We recall from Section~\ref{sec:model}, that for materials at a high
value of the atomic number ${\cal Z}$, when imaged at low resolutions,
the unknown functions are learnt as per the prescription of the 1st
model. See Equation~\ref{eqn:highzlowz} for definition of ``high'' and
``low'' ${\cal Z}$ materials. In this connection we mention that for
the low resolution imaging techniques (such as imaging in X-rays), the
smallest resolved length in the observed images is $\omega$=1.33$\mu$m
while for the illustration of the high resolution technique (such as
imaging with Back Scattered Electrons), a much finer resolution of
$\omega$=9 nm, is considered. Materials for which ${\cal Z}$ is low,
when imaged at low resolutions, are modelled using the 2nd model. Any
material that is imaged at high resolution, i.e. for small values of
$\omega$, will be modelled using the 3rd-model. The true density and
correction function are then used in the model appropriate equation
(Equation~\ref{eqn:highz_chain} for the 3rd model and
Equation~\ref{eqn:highz_chain} for the other 2 models) to generate the
sequential projections of the convolution $\rho\ast\eta$ to give us
the simulated image data
$\{\tilde{I_i^{(k)}}\}_{i=1;\:k=1}^{N_{{\textrm{data}}},\:N_{{\textrm{eng}}}}$. The simulated
images are produced with noise of 5$\%$ to 3$\%$,
i.e. $\sigma_i^{(k)}$=0.05 to 0.03 times ${\tilde I}_i^{(k)}$. 

We start MCMC chains with a seed density $\xi_{i,seed}^{(k)}$ in the
$ik^{th}$ voxel defined as $\xi_{i,seed}^{(k)}={\tilde
  I}_i^{(k)}/\Phi_{seed}^{(k)}$, where the starting correction
function $\{\eta_{seed}^{(k)}\}_{k=1}^{N_{{\textrm{eng}}}}$ helps define
$\Phi_{seed}^{(k)} =
\sum_{j=1}^{k}\eta_{seed}^{(j)}(h^{(j)}-h^{(j-1)})$.  The initial seed
for the correction function is chosen to be motivated by the forms
suggested in the literature; in fact, we choose the initial correction
function to be described by a half-normal distribution, with a
standard deviation of 5$\mu$m and mean of 5$\mu$m. We use
adaptive Metropolis within Gibbs for our inference, as discussed above
in Section~\ref{sec:inference}.

\subsection{Low resolution, high-${\cal Z}$ system}
\noindent
These illustrations are performed with intrinsic densities that are
sparse (Sample~II-CuW) and dense (Sample~I-CuW). The idea behind the
modelling in this case is that the cross-sectional area of the surface
of the interaction-volume at any $E$, is less than the cross-sectional
area of the voxel on this $Z$=0 surface, rendering the inversion the
simplest out of the three illustrations we discuss. An example
resolution of $\omega$=1.33 $\mu$m of such low-resolution imaging
techniques, allows for 15 beam pointings over an interval of 20
$\mu$m, from $x$=-10$\mu$m to $x$=10$\mu$m at given $y$, at spatial
intervals where $\omega$. Again, at a given $x$, there are 15 beam
pointings, at steps of $\omega$, from $y$=-10$\mu$m to
$y$=10$\mu$m. $N_{{\textrm{eng}}}$ was chosen as 18. We invert the image data
thus generated, both when we consider a distribution-free model for
$\eta(z)$ as well as the parametric model.

\begin{figure*}
     \begin{center}
  {$\begin{array}{c c}
       \includegraphics[width=6cm]{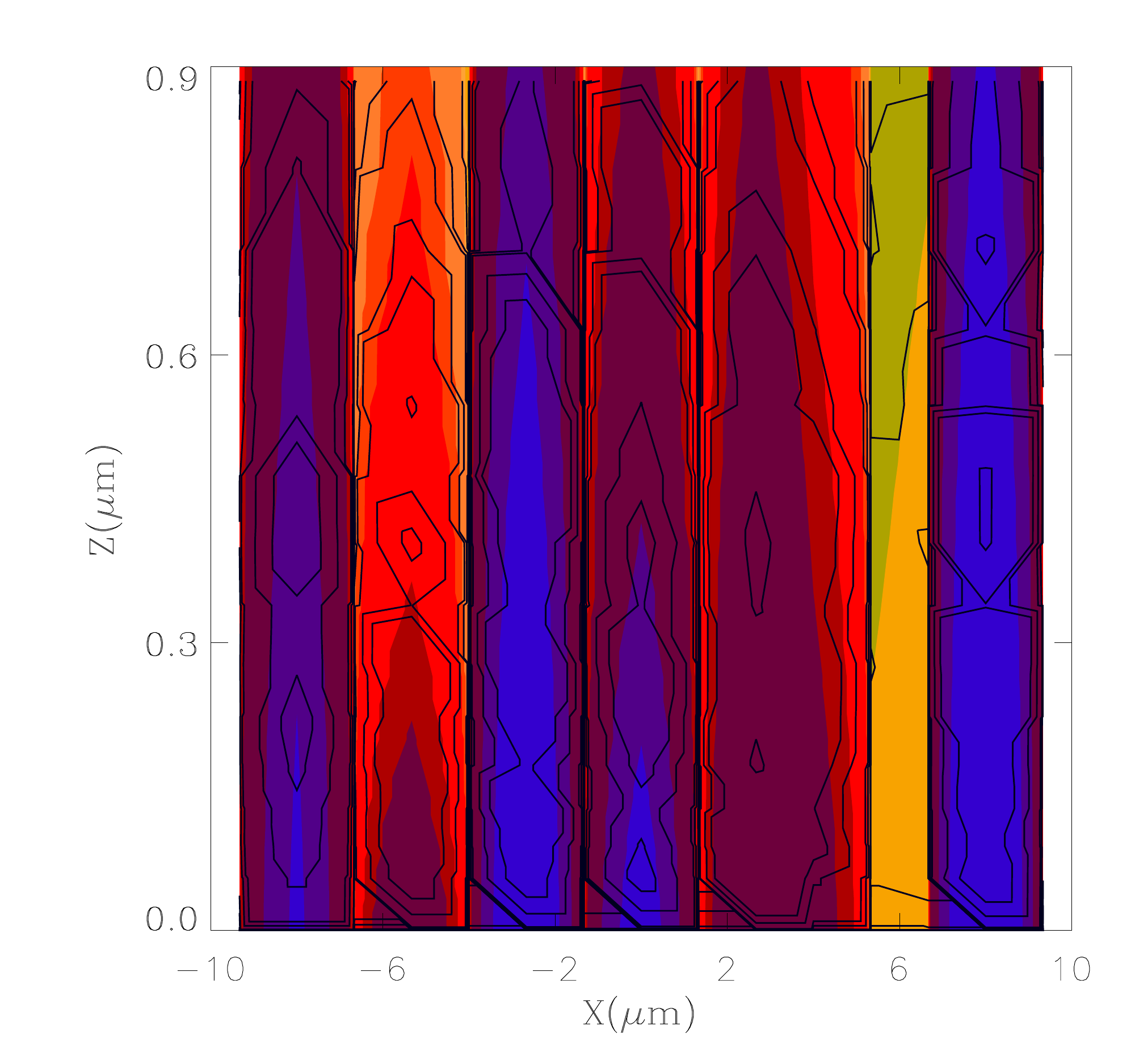}
       \includegraphics[width=6cm]{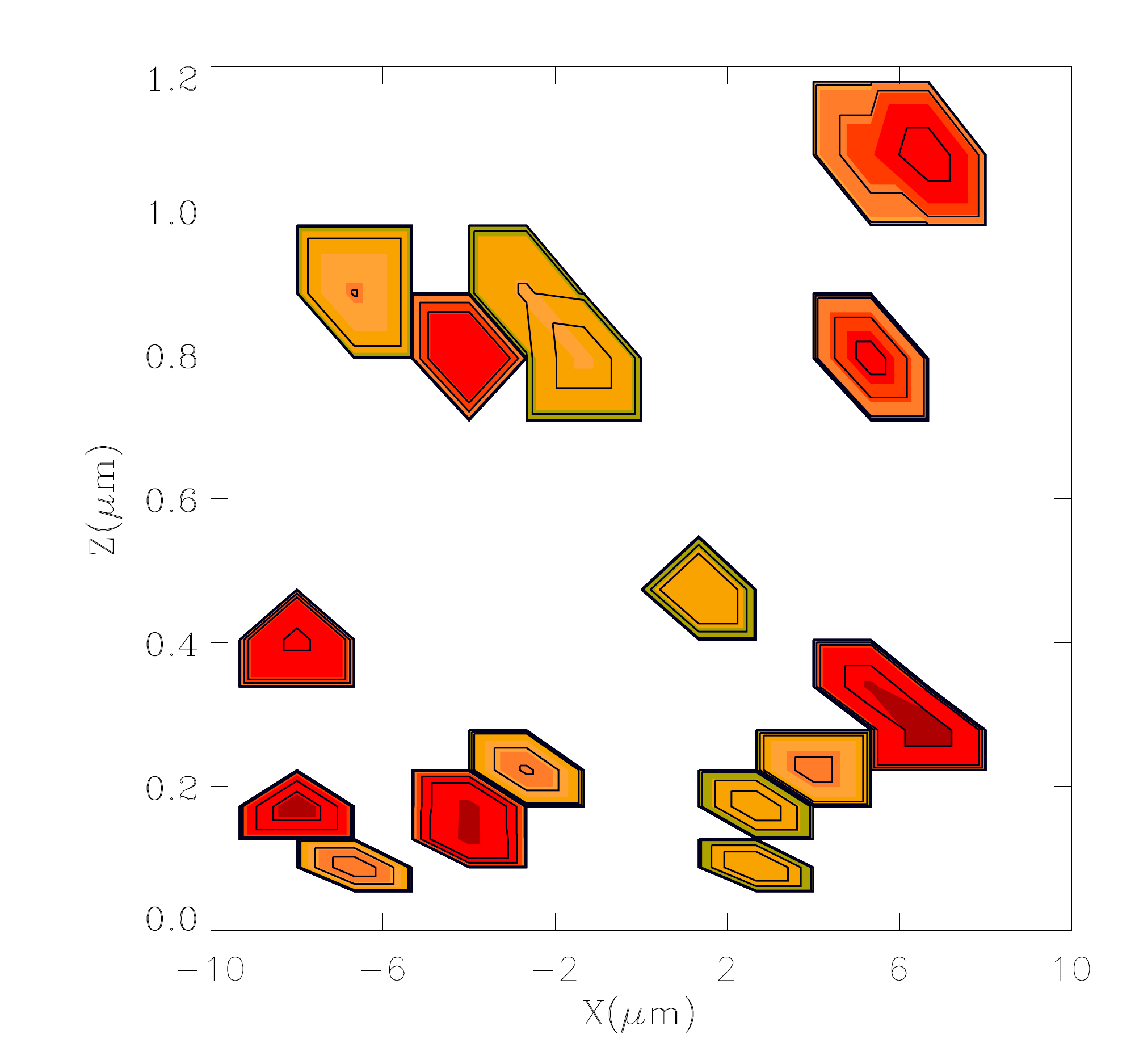}
  \end{array}$}
     \end{center}
\caption{Estimated material density parameters over-plotted in the
  $X-Z$ plane, $\forall y$, in solid contour lines, over the true
  density (in filled coloured contours), given simulated images of
  Copper-Tungsten alloy - Sample~I-CuW (left) and Sample~II-CuW
  (right). The simulated image data corresponds to a low-resolution
  imaging technique (such as when the system is imaged by an SEM in
  X-rays), at different beam energy values. The material density
  estimate at the medial level of the posterior probability is
  plotted. The parametric model for the correction function was
  used for these runs. }
\label{fig:colour}
\end{figure*}

Figure~\ref{fig:colour} represents the estimated density functions for
the two simulated samples, represented as contour plots in the $X-Z$
plane, $\forall y$, superimposed over the respective true densities
which are shown as filled contour plots, when the parametric
model for $\eta(z)$ is used. We set $\sigma_i^{(k)}=0.05{\tilde
  I}_i^{(k)}$. The panel on the left describes a true density function
that is dense in its native space while on the right, the true density
function is sparse.

Figure~\ref{fig:1d} represents results of implementation of the
distribution-free models for the kernel, using
$\sigma_i^{(k)}=0.03{\tilde I}_i^{(k)}$. This figure includes the plot
of ${\cal C}(\rho\ast\eta)_i^{(k)}$ as a function of the pixel location
along the $X$ axis, over all $k=1,\ldots,N_{{\textrm{eng}}}$. 
This data is over-plotted on the image data ${\tilde
  I}_i^{(k)}$, plotted against $\lambda_i$ for all $k$. Similarly, the
estimated ${\hat\eta}(z)$ and true correction function are also
compared. Trace plots for the chains are also included.

It is to be noted that the definition of the likelihood in terms of
the mismatch between data and the projection of the learnt density and
kernel parameters, i.e. ${\cal C}({\hat\rho}\ast{\hat\eta})_i^{(k)}$,
compels ${\cal C}({\hat\rho}\ast{\hat\eta})_{\cdot}^{(\cdot)}$ and
${\tilde{I}}_{\cdot}^{(\cdot)}$ to coincide, even if ${\hat\rho}({\bf
  x})$ is poorly estimated, though such a comparison serves as a
zeroth order check on the involved coding.

\begin{figure*}[!h]
\begin{center}{
\includegraphics[width=13cm]{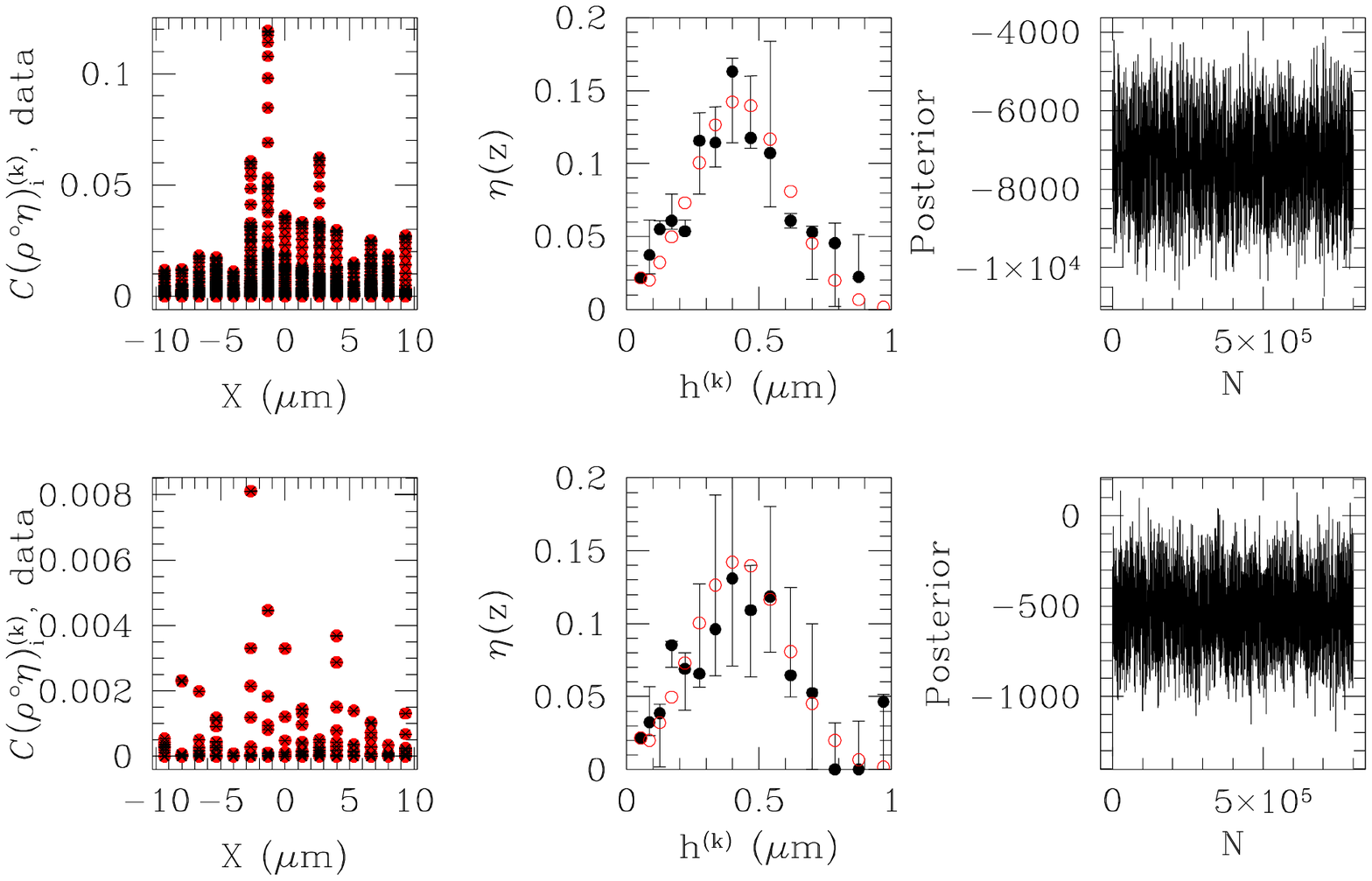}}
\end{center}
\vspace*{-3in}
\caption{{\it Left:} - projection ${\cal
    C}({\hat\rho}\ast{\hat\eta})_i^{(k)}$ plotted in black, in model
  units, as a function of beam location on the $X$-axis, (where the
  beam pointing index $i=1,\ldots, 225$), over the image data (in
  red), for all beam energy values $\epsilon_k=2+k$ (in kiloVolts),
  $k=1,\ldots,18$, for Sample~I-CuW (upper panel) and Sample~II-CuW
  with the relatively sparser material density (lower). {\it Middle:}
  estimated correction function, in black, superimposed on the true
  $\eta^{(\cdot)}$ for the two samples (in red). The error bars in the
  above plots, as well as in all plots that follow, correspond to
  95$\%$ highest probability density credible regions, while the
  medial level of the posterior is marked by a symbol (filled
  circle). The trace of the joint posterior probability density of the
  unknown parameters, given the image data is shown in the right
  panels. The distribution-free model for the correction function was
  implemented. }
\label{fig:1d}
\end{figure*}

\subsection{Low resolution, low-${\cal Z}$ system}
\noindent
Distinguished from the last section, in this section we deal with the
case of a low atomic number material imaged at coarse or low imaging
resolution. In this case, the surface cross-sectional area of an
interaction-volume exceeds that of a voxel only for $E=\epsilon_{k} \leq
\epsilon_{k_{in}}$, for $k=1,\ldots,k_{in},k_{in+1},\ldots,N_{{\textrm{eng}}}$.
The illustration of this case is discussed here for Sample~II-NiAl,
that is an alloy of Nickel, Aluminium, Tallium, Rhenium; see
\ctn{dsouza}. 
The average atomic weight of our
simulated sample is such that $k_{in}$=13 but the extent of interaction-volumes attained at $E=\epsilon_{k}$ (=$k$+2 kiloVolts) for $k=14,\ldots,18$, is in excess of the image resolution 1.33 $\mu$m.
All other details are as in the previous illustration.

Figure~\ref{fig:Al} 
depicts the learnt density structure; visualisation to this effect is
provided for the Ni-Al alloy sample in terms of the plots of
$\xi_i^{(k)}$ against $h^{(k)}$, at chosen values of $X$ and $Y$
inside the sample. Since there are 225 voxels used in these exercises,
the depiction of the density at each of these voxels cannot be
accommodated in the paper, but Figure~\ref{fig:Al} shows the density
structure at values of beam pointing index $i$ corresponding to 15
values of $Y$, at one
value of the $X$ (=6.67$\mu$m).

\begin{figure*}
\centerline{
\includegraphics[width=0.99\hsize]{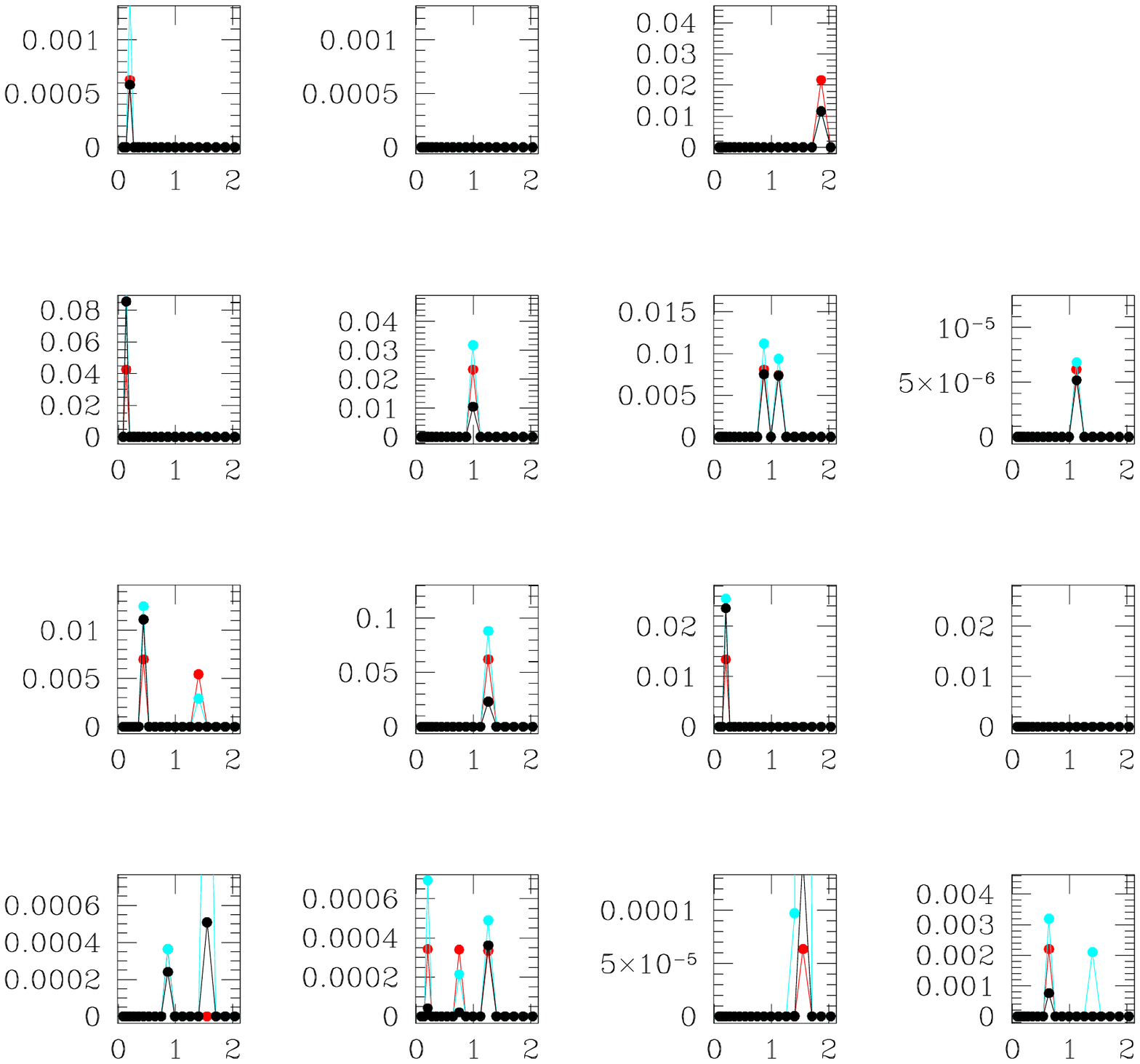}}
\vspace*{-2in}
\caption{For the simulated (sparse and low-${\cal Z}$)
  Sample~II-NiAl, the estimated density parameters $\xi_i^{(k)}$
  plotted in model units as a function of $h^{(k)}$ (in $\mu$m), at
  $X=6.67$$\mu$m and
  distinct values of the beam pointing locations along the $Y$-axis, in
  the range of -10$\mu$m to 10$\mu$m, at intervals of 1.33$\mu$m where
  the panel in the lower left-hand corner is at the smallest value of $Y$ and that in the upper right-hand most panel is at the highest value of $Y$. The true variation of the learnt material density with
  depth is shown in red in each panel. This is compared to the material density
  estimated at the lower and upper bounds of the 95$\%$ HPD credible
  region, as plotted in cyan and black respectively (in model
  units). The parametric model for $\eta^{(\cdot)}$ was used. }
\label{fig:Al}
\end{figure*}


\subsection{High resolution}
\noindent
The challenge posed by the high resolution
($\omega\lesssim$10 nm=0.01 $\mu$m) of imaging, to our modelling is
logistical. This logistical problem lies in fact that for
$\omega\lesssim$10 nm the run-time involved in the reconstruction of
the density over the length scales of $\sim$1 $\mu$m, is prohibitive;
for an image resolution of $\omega$=10 nm, there are 100 voxels
included over a length interval of 1 $\mu$m, while at example lower
resoutions relevant to the aforementioned illustrations there would be
${\mbox{int}}(100/15)$. 

In light of this problem, for a high resolution of $\omega$=2 nm that
we work with, we choose 20 voxels (each of cross-sectional size 2 nm),
across the $X$ range of -20 nm to 20 nm. A similar range in $Y$ is
considered. The simulated system is imaged at 10 energy values
$\epsilon_k=k+1.5$kV, $k=1,\ldots,10$. The atomic number parameter
${\cal Z}$ of the used material is such that the radius
$R0^{(1)}=h^{(1)}$ of the interaction-volume attained at $k=1$ is
about 21 nm, so that all the studied voxels live inside this and all
other (larger) interaction-volumes. It is possible to run multiple
parallel chains with data from distinct parts of the observed image,
in order to cover this available image. That way, the material density
function of the whole sample will be learnt.

The plots of the learnt density for the intrinsically dense and sparse
true density functions are shown as a function of depth $Z$, for all
$X$ and $Y$, in Fig~\ref{fig:sem}, which display results of runs done
with the distribution-free as well as parametric models for the
correction function (using $\sigma_i^{(k)}=0.05{\tilde I}_i^{(k)}$,
$i=1,\ldots,N_{{\textrm{data}}}$, $k=1,\ldots,N_{{\textrm{eng}}}$). The multiscale structure
of the heterogeneity in these simulated density functions is brought
out by presenting the logarithms of the estimated densities parameters
alongside the logarithms of the chosen or true density parameters.

\begin{figure*}[!h]
\vspace*{-4in}
\begin{center}
{\includegraphics[width=14cm]{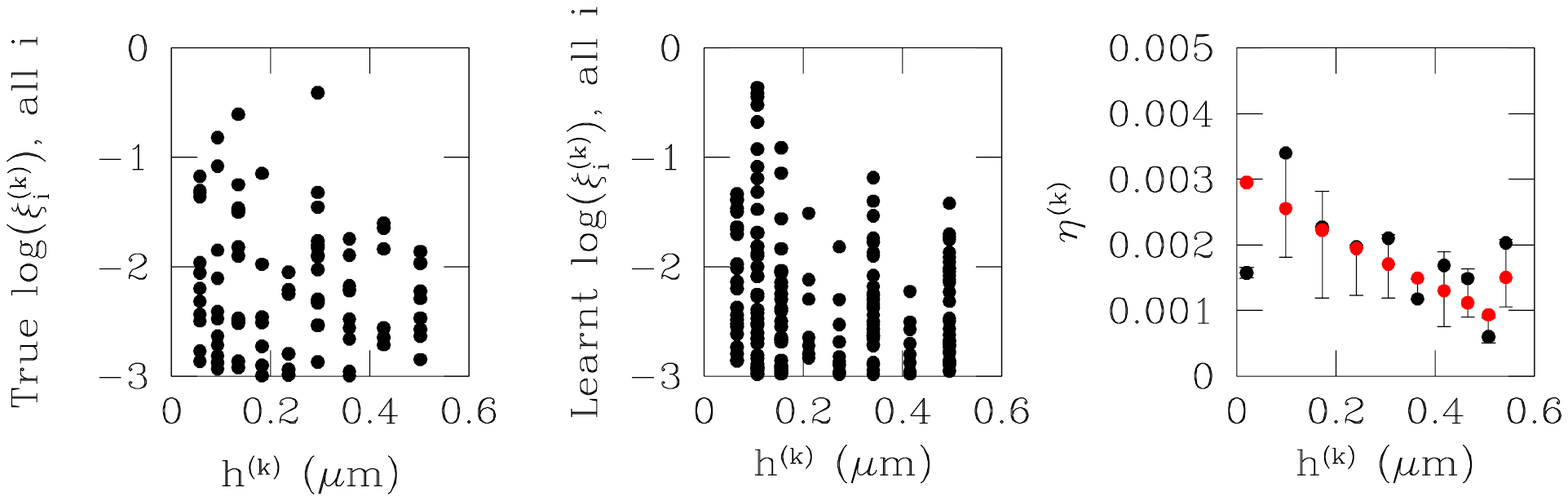}}
\end{center}
\vspace*{-4.5in}
\begin{center}
{\includegraphics[width=14cm]{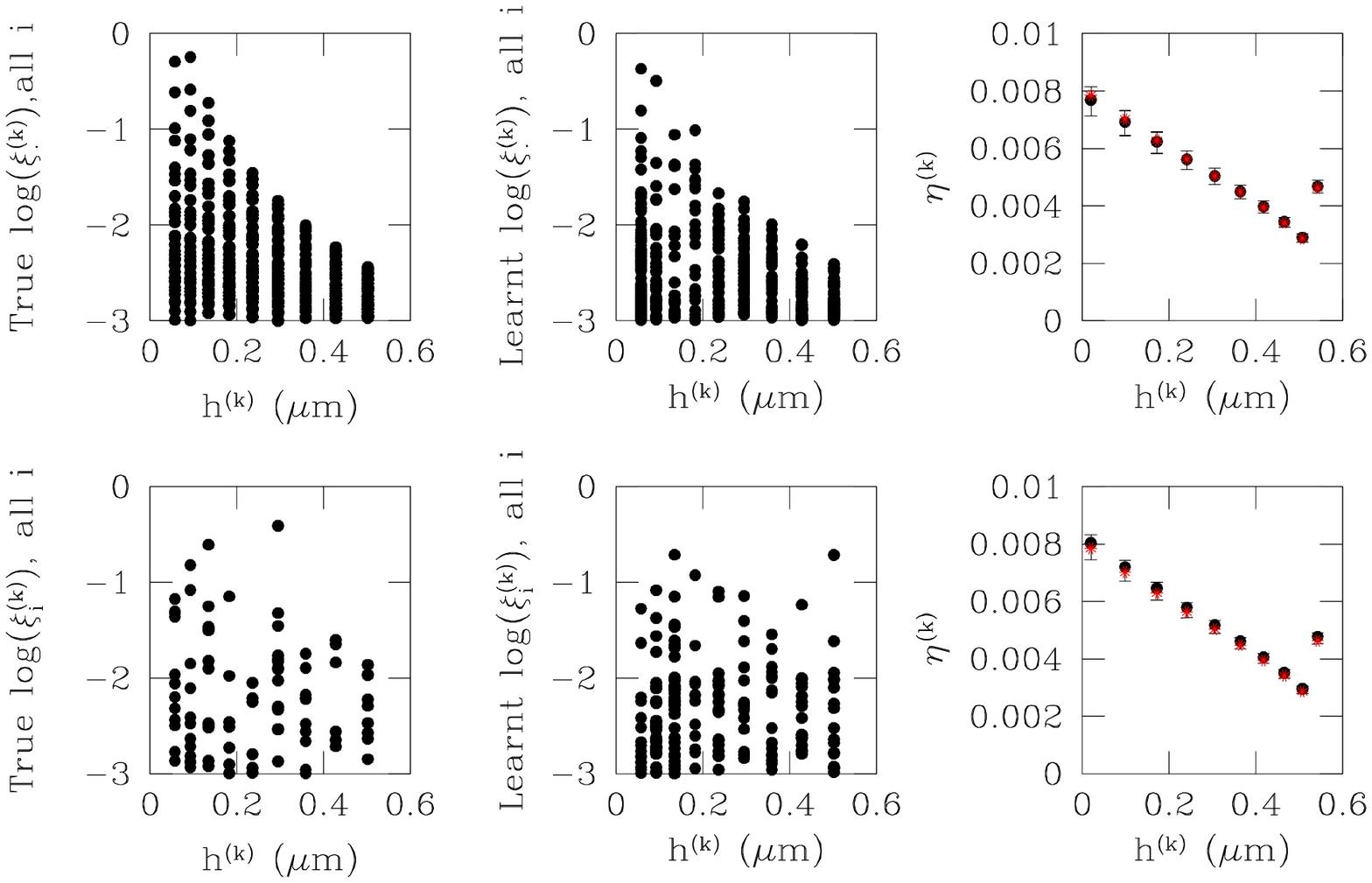}}
\end{center}
\vspace*{-1.5in}
\caption{{\it Top row:} results of inversion of simulated image data
  of the intrinsically sparse low-${\cal Z}$ Sample~II-NiAl when a
  distribution-free model for the correction function is used. The
  true material density structure is presented in the panels on the
  left, as a function of $z$, for all $X, Y$, in model units. The
  learnt material density, at the medial level of the posterior, is
  similarly presented in the middle panels while the learnt
  $\eta^{(k)}$ over 95$\%$ HPD credible region is shown in black as a
  function of $h^{(k)}$ in the right panels, in model units,
  superimposed on the true kernel parameters in red. {\it Middle row:}
  the intrinsically dense, true and learnt (at the posterior median)
  material density of simulated sample Sample~I-NiAl in the left and
  middle panels respectively with the learnt $\eta^{(\cdot)}$. A
  parametric model for the correction function is employed. {\it
    Bottom row:} results from implementing simulated images of the
  sparse Sample~II-NiAl when a parametric model for the correction
  function is employed.  }
\label{fig:sem}
\end{figure*}

\section{Application to the analysis of real SEM image data of Nickel and Silver nano-composite}
\label{sec:real}
\noindent
In this section we discuss the application of the methodology advanced
above towards the learning of the 3-D material density and the
microscopy correction function (or kernel) by inverting 11 images
taken with a real SEM, in a kind of radiation called Back Scattered
Electron (BSE). The imaged sample is a brick of Nickel (Ni) and Silver
(Ag) nanoparticles, taken with the SEM (the Leica, Stereoscan 430
brand), at 11 distinct values of the beam energy parameter $E$ such
that the values that $E$ takes are $10,11,\ldots,20$kV. The brick of
nanoparticles was prepared by the drop-cast method in the
laboratory. The resolution of BSE imaging with the used SEM is
$\approx$50 nm=0.05 $\mu$m, i.e. the smallest length over which
sub-structure is depicted in the image is about 50 nm. While this
resolution can be coarser than that for images taken in the radiation of
Secondary Electrons \ctp{reed}, an image taken in Secondary Electrons
will however not bear information coming from the bulk of the
material. Hence the motivation behind using BSE image data.

We sample two distinct areas from each of the 2-D images taken at
values $10,11,\ldots,20$kV of $E$, resulting in two different image
data sets. The first data $D_1$ comprises a square area of size
101pixels$\times$101pixels, with this square seated in a different
location on an image than the smaller square that contain the pixels
that define the image data set $D_2$. This second data $D_2$ comprises
a much smaller square of area 41pixels$\times$41pixels. Given the
imaging resolution length of 0.05$\mu$m, i.e. given that the width of
each square pixel is 0.05$\mu$m, in image data set $D_1$ the 101
pixels along either the $X$ or $Y$-axes are set to accommodate the
length interval [-2.5 $\mu$m, 2.5 $\mu$m] while in data $D_2$, the
pixels occupy the interval [-1 $\mu$m, 1 $\mu$m]. Then in data set
$D_1$, at a given value of $E$, image datum from a total of
100$^2$=10,000 pixels are recorded, i.e. $N_{{\textrm{data}}}$ for this data set
is 10,000. For $D_2$, $N_{{\textrm{data}}}$=41$^2$=1,764. Both $D_1$ and $D_2$
comprise 11 sets of BSE image data $\{{\tilde
  I}_i^{(k)}\}_{i=1}^{N_{{\textrm{data}}}}$, with each such set imaged with a
beam of energy $\epsilon_k$, $k=1,\ldots,11$ such that
$\epsilon_k=k+9$ kiloVolts. For this material, the mean of the atomic number, atomic mass and physical density in gm cm$^{-3}$ yields average atomic
number parameter ${\cal Z}$ of 37.5, and the numerical values of the
parameters $A$ and $d$ in Equation~\ref{eqn:kanaya} as 83.28 and
about 9.7 respectively; then the smallest interacion-volume attained
at $E=\epsilon_1$ is $R0^{(1)}=h^{(1)}\approx$0.44 $\mu$m and the largest at
$E=\epsilon_{10}$ is about 1.4 $\mu$m.

The learning of the material density of a nanostructure is very useful
for device engineers engaged in employing such structures in the
realisation of electronic devices \ctp{shashinano}; lack of
consistency among measured electronic characteristics of the formed
devices is tantamount to deviation from a standardised device
behaviour and such can be predicted if heterogeneity in the depth
distribution of nanoparticles is identified. Only upon the receipt of
quantified information about the latter, is the device engineer able
to motivate adequate steps to remedy the device realisation. In
addition, such information holds potential to shed light on the
physics of interactions between nanoparticles.

If the calibration of the intensity $I_i^{(k)}$ of the image datum is
available in physical units (of surface density of the BSE), then we
could express the measured image data - as manifest in the recorded
image - in relevant physical units. In that case, the learnt density
could be immediately expressed in physical units. However, such a
calibration is not available to us. In fact, in this method, the
learnt density is scaled in the following way. The measured value of
$\eta^{(1)}$ from microscopy theory \ctp{corr_epm} for this material
suggests using the normalisation factor for $\eta^{(\cdot)}$ is
$\approx {\hat\eta}^{(1)}/0.325$, so that the $\eta^{(\cdot)}$
parameters, thus normalised, are in the physical units of $\mu$m.
(The arithmetic mean of the atomistic parameters of Ni and Ag yield
the value of 0.325). This in turn would imply that the material
density parameters are each in the physical unit of $\mu$m$^{-3}$
since the product of the density and kernel parameters appears in the
projection which is measured in physical units of ``per unit area'',
i.e. in ``per length unit$^2$'' or ($\mu$m)$^{-2}$. Here
${\hat\eta}^{(1)}$ is the learnt value of the kernel in the first
$Z$-bin. Physically speaking, here we learn the number density in each
voxel.

\begin{figure*}[!t]
     \begin{center}
  {$\begin{array}{c c}
       \includegraphics[width=5cm]{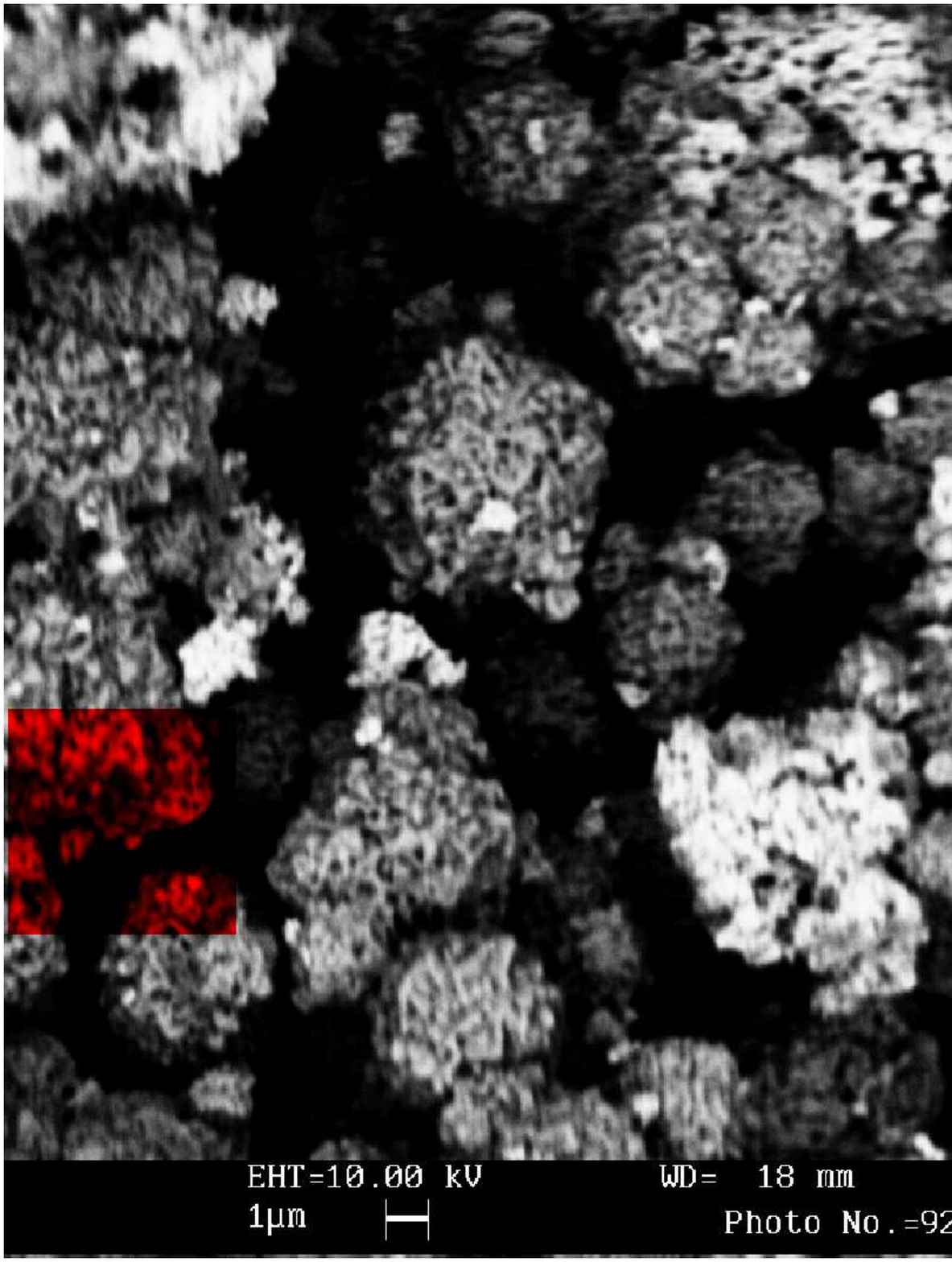} &
\hspace{.5cm}
       \includegraphics[width=5cm]{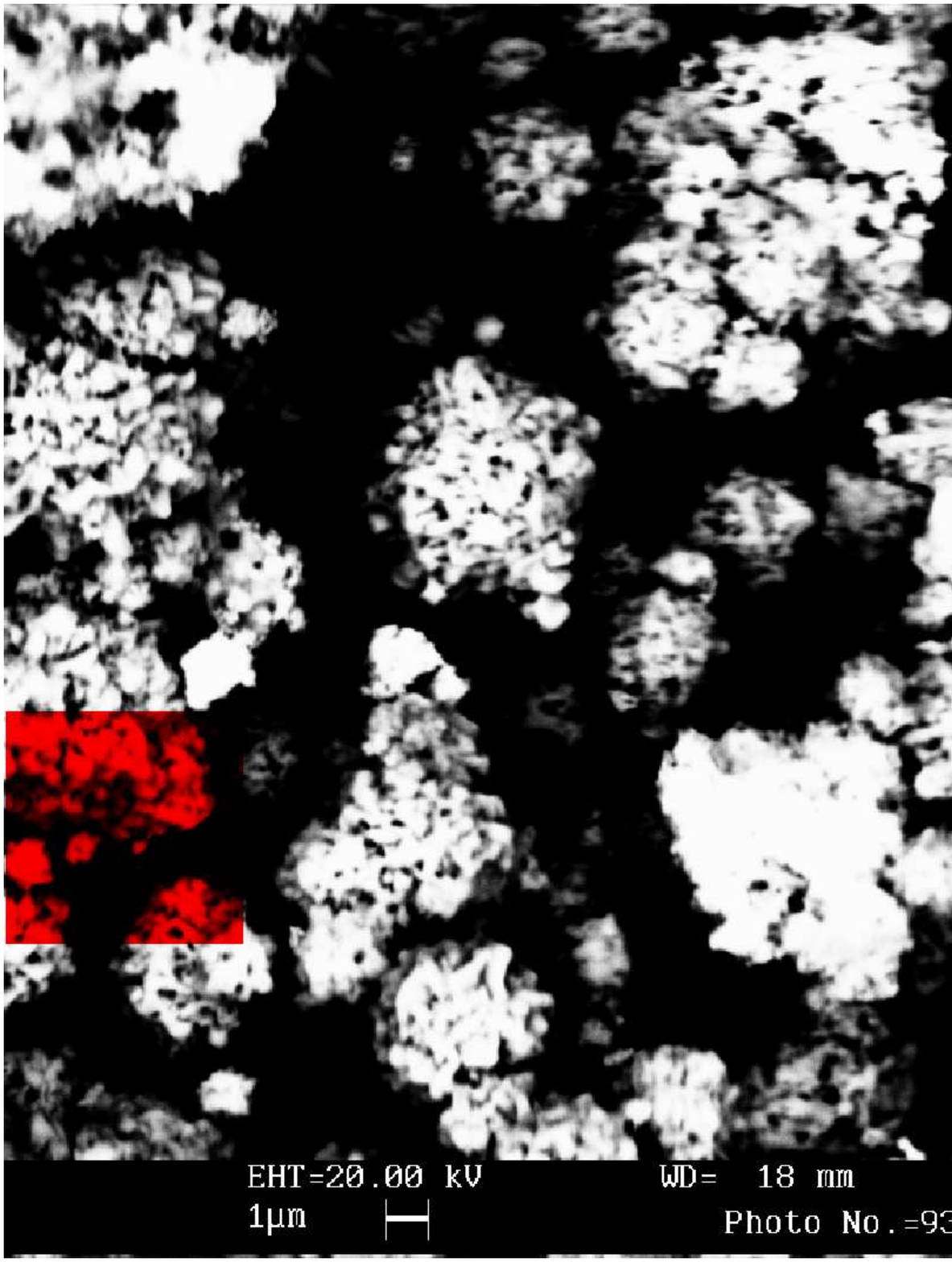} 
  \end{array}$}
     \end{center}
\vspace*{-1in}
     \begin{center}
  {$\begin{array}{c c}
       \includegraphics[width=5.5cm,height=7.5cm]{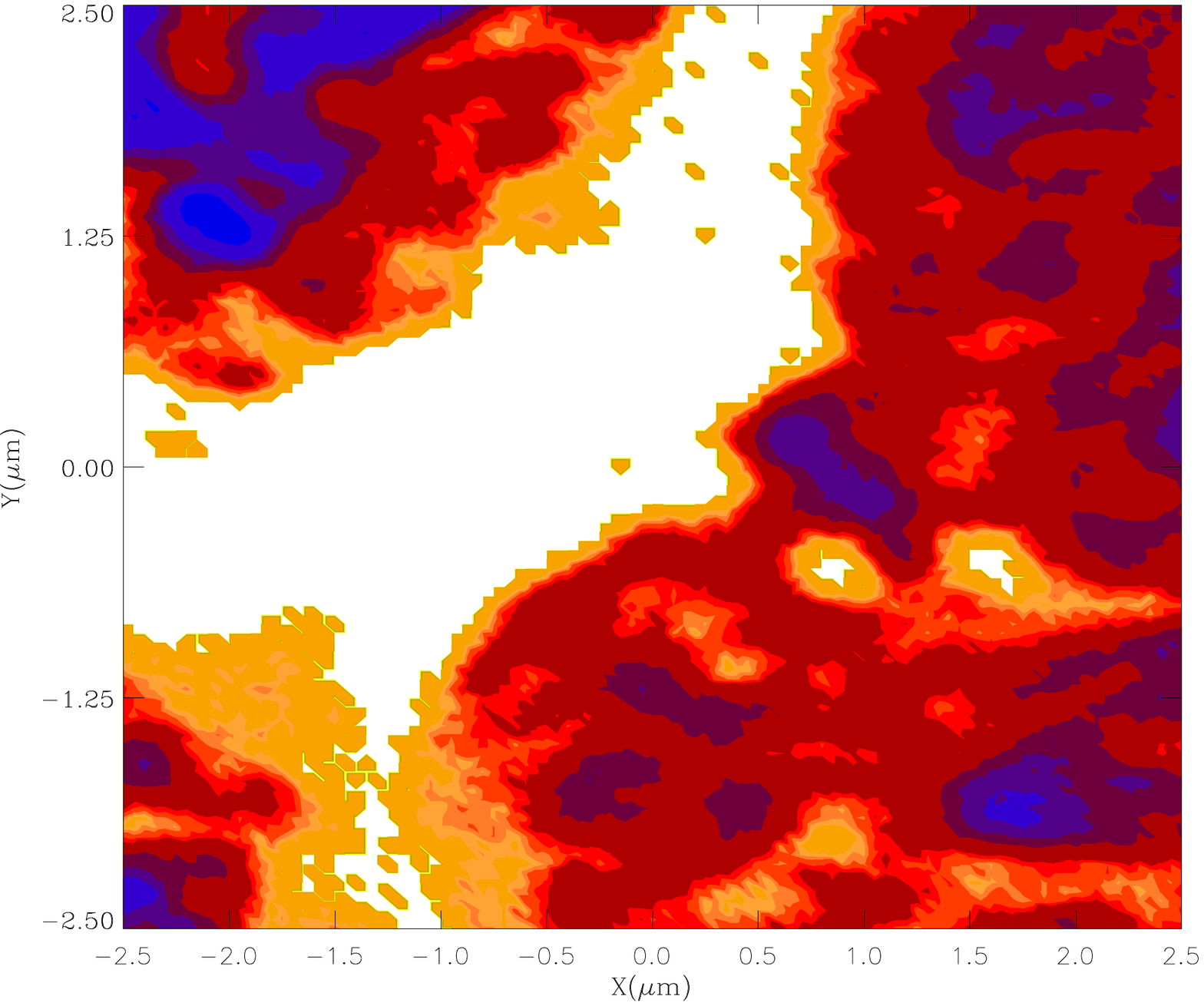} &
       \includegraphics[width=5.5cm,height=7.5cm]{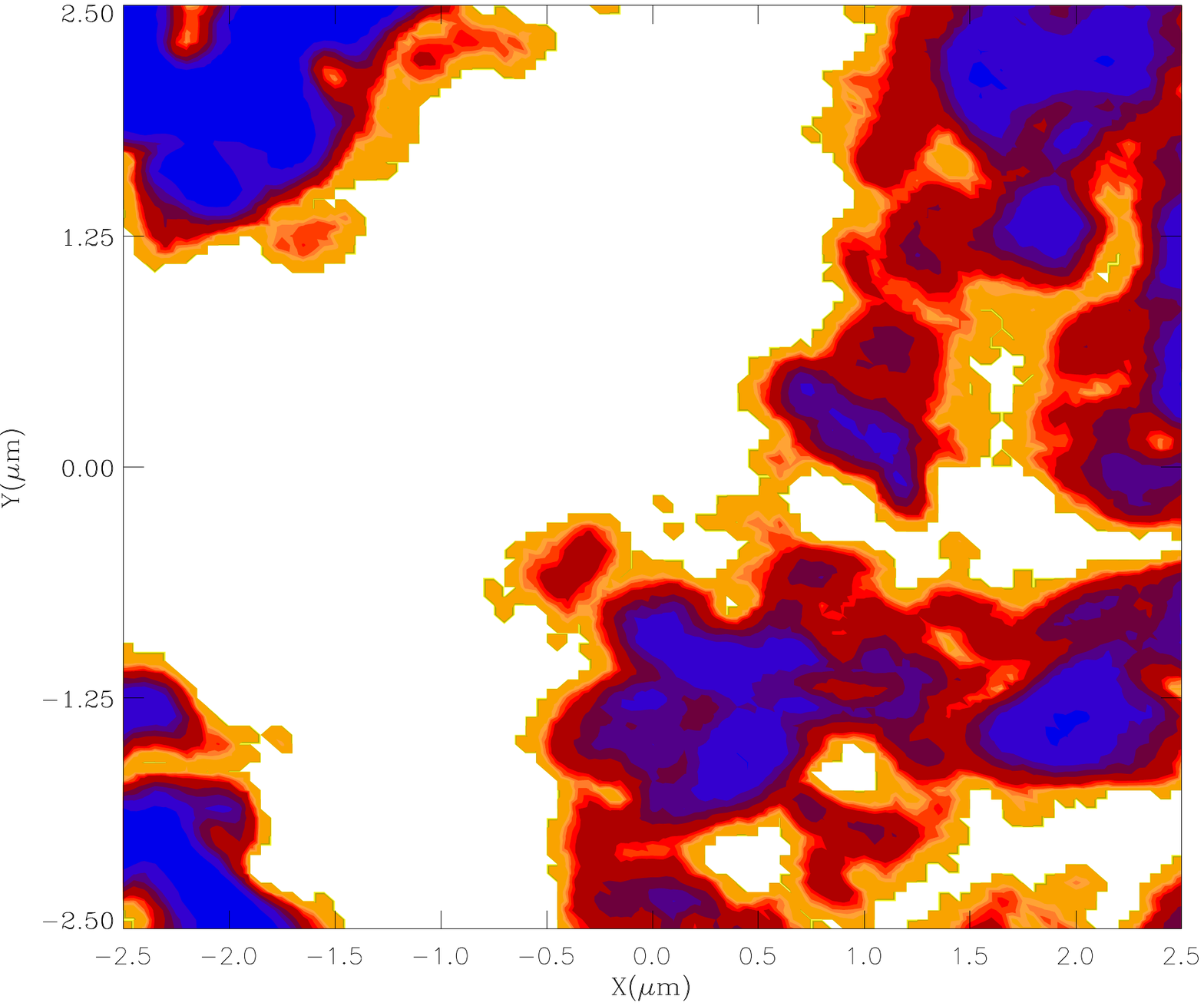} 
  \end{array}$}
     \end{center}
\caption{\small{The top panels display two of the 11 images of the
    prepared blend of Nickel and Silver nanoparticles, taken with an
    SEM, at beam energy values
    of 20 kV (right) and 10 kV (left). A 5 $\mu$m$\times$ 5 $\mu$m
    area was identified in each of the 11 images taken in the kind of
    radiation called Back Scattered Electrons (BSE), to form the data
    $D_1$. The colourised squares in red in each image represent the
    areas that contribute to $D_1$ at energies 20 kV and 10 kV. The
    distribution of the measured image data over these square areas in
    red (rotated clockwise by 90$^\circ$), for beam energies of 10 kV
    and 20 kV are shown in the left and right panels repectively, on the lower
    row.}}
\label{fig:images}
\end{figure*}

The configuration that the constituent nanoparticles are expected to
relax into, is due to several factors, including contribution from the
surface effects - the surfaces of nanoparticles are active - this
encourages interactions, resulting in a clustered
configuration. Additionally, gravitational forces that the
nanoparticle aggregates sediment under, are also active\footnotemark,
with the nanoparticle diameter responsible for determining the
relative importance of the different physical influences
\ctp{handbook}. Thus, in general, along the $Z$-axis, we expect a
clustered configuration, embedded within layers. This is what we see
in the representation of the learnt density in the $Y=0$ plane, when
the image data $D_1$ and $D_2$ are inverted; see
Figure~\ref{fig:real_den}. In Figure~\ref{fig:real_corr} we present
the learnt kernel parameters; the kernels learnt by using data from
two distinct parts of the image are expected to overlap given that in
our model, the kernel function is a function of $Z$ alone with no
dependence on $X$ and $Y$.  Such is recovered using data $D_1$ and
$D_2$, as shown in the left panel of this figure. The sequential
projections of the convolution of the learnt density and learnt
correction function, onto the image space, using the data set $D_1$ is
overplotted on the image data in Figure~\ref{fig:real_corr}. The
parametric model for the correction function is implemented here.
\footnotetext{Other relevant effects include viscosity and Brownian
  motion.}

\begin{figure*}
\vspace*{-2in}
     \begin{center}
  {$\begin{array}{c c}
       \includegraphics[width=6cm]{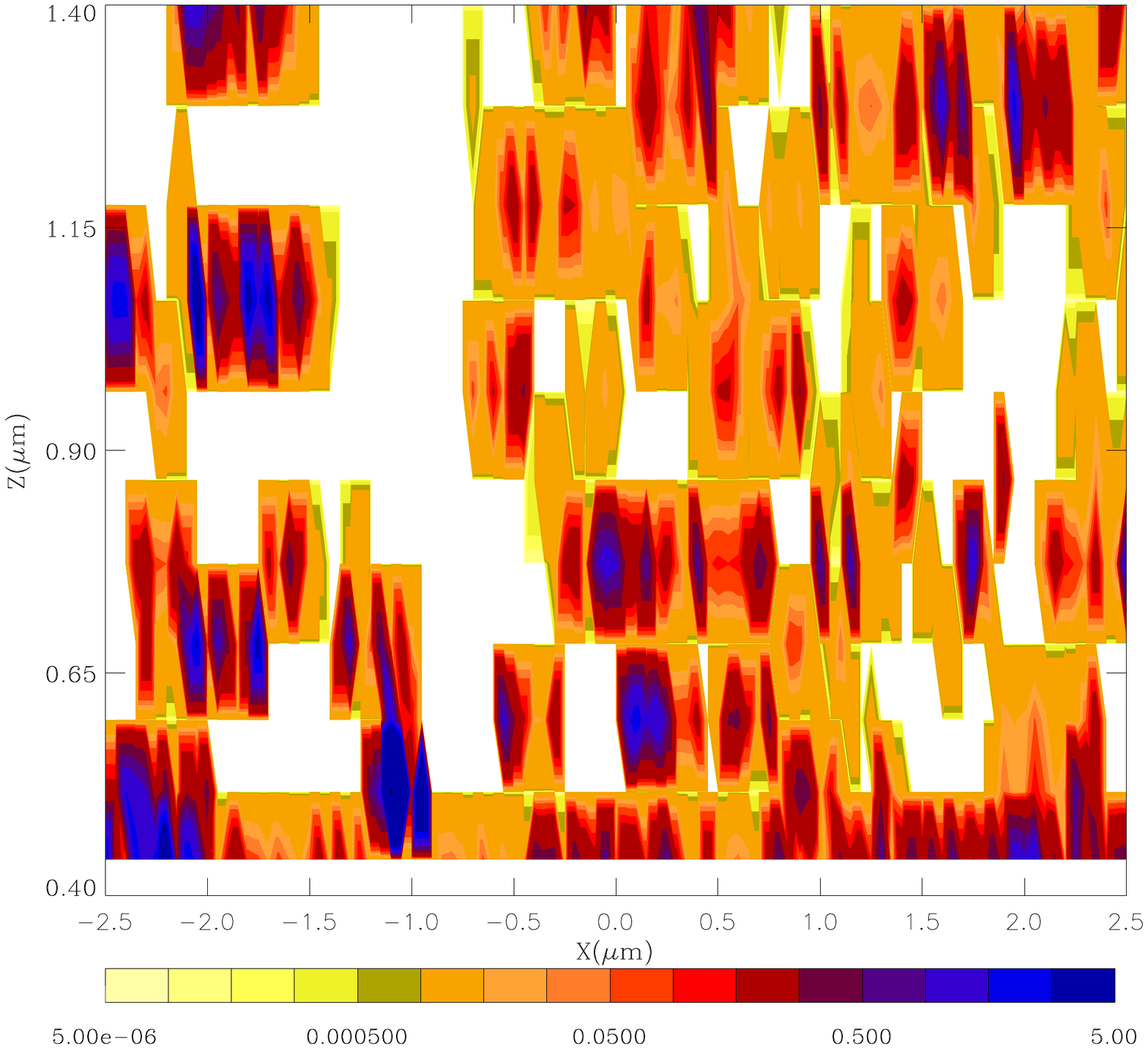}
       \includegraphics[width=6cm]{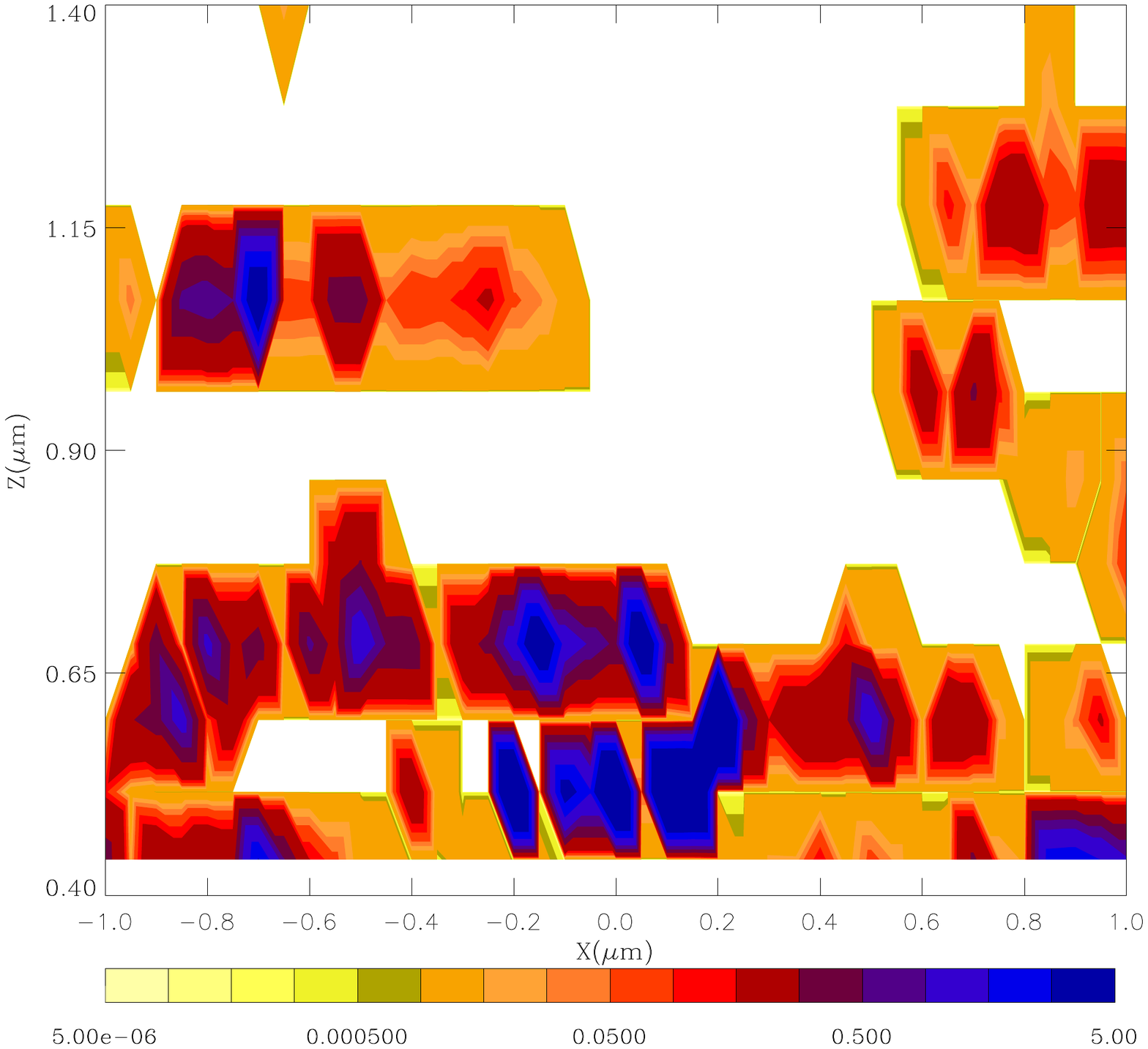}
  \end{array}$}
     \end{center}
\caption{\small{The left and right panels show the slice of the learnt
    three dimensional material density in the $Y=0$ plane from
    inversion of image data $D_1$ and $D_2$ respectively. In fact, to
    construct these figures we use the density parameter value that is
    the median of the posterior of the parameter given the image data. The
    data are obtained by imaging a brick of Nickel and Silver
    nanoparticles in a radiation called Back Scattered Electrons, at
    resolution of 0.05$\mu$m, at 11 different values of $E$, from 10 kV,
    $\ldots$, 20 kV. $D_1$ and $D_2$ comprise 101$\times$101 pixels
    and 41$\times$41 pixels respectively. The labels on the colour-bar
    are number density values, in physical units (of $\mu$m$^{-3}$).}}
\label{fig:real_den}
\end{figure*}

\subsection{MCMC convergence diagnostics}
\noindent
In this section we include various diagnostics of an MCMC chain that
was run until convergence, using the image data $D_1$. These include
trace of the likelihood (Figure~\ref{fig:inf2}), 
and
histograms of multiple learnt parameters - $\xi_{50}^{(1)}$,
$\eta^{(1)}$ - from 1000 steps, in two distinct parts of the chain,
namely, for step number $N\in$[1599001,1600000] and
$N\in$[799001,800000], respectively (Figure~\ref{fig:real_corr2}). The
histograms of the likelihood over these two separate parts of the
chain are also presented in this figure.

\begin{figure*}
\vspace*{-6in}
     \begin{center}
  {
  \hspace*{-10in}\includegraphics[width=37cm, height=43cm]{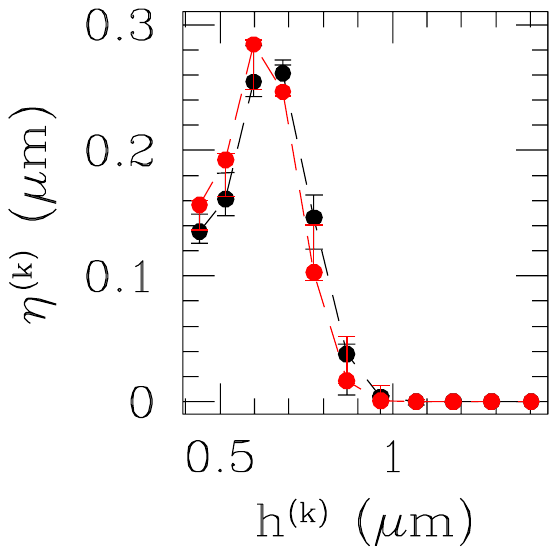} 
 }
     \end{center}
\vspace*{-8in}
     \begin{center}
  {
  \includegraphics[width=8cm, height=10.2cm]{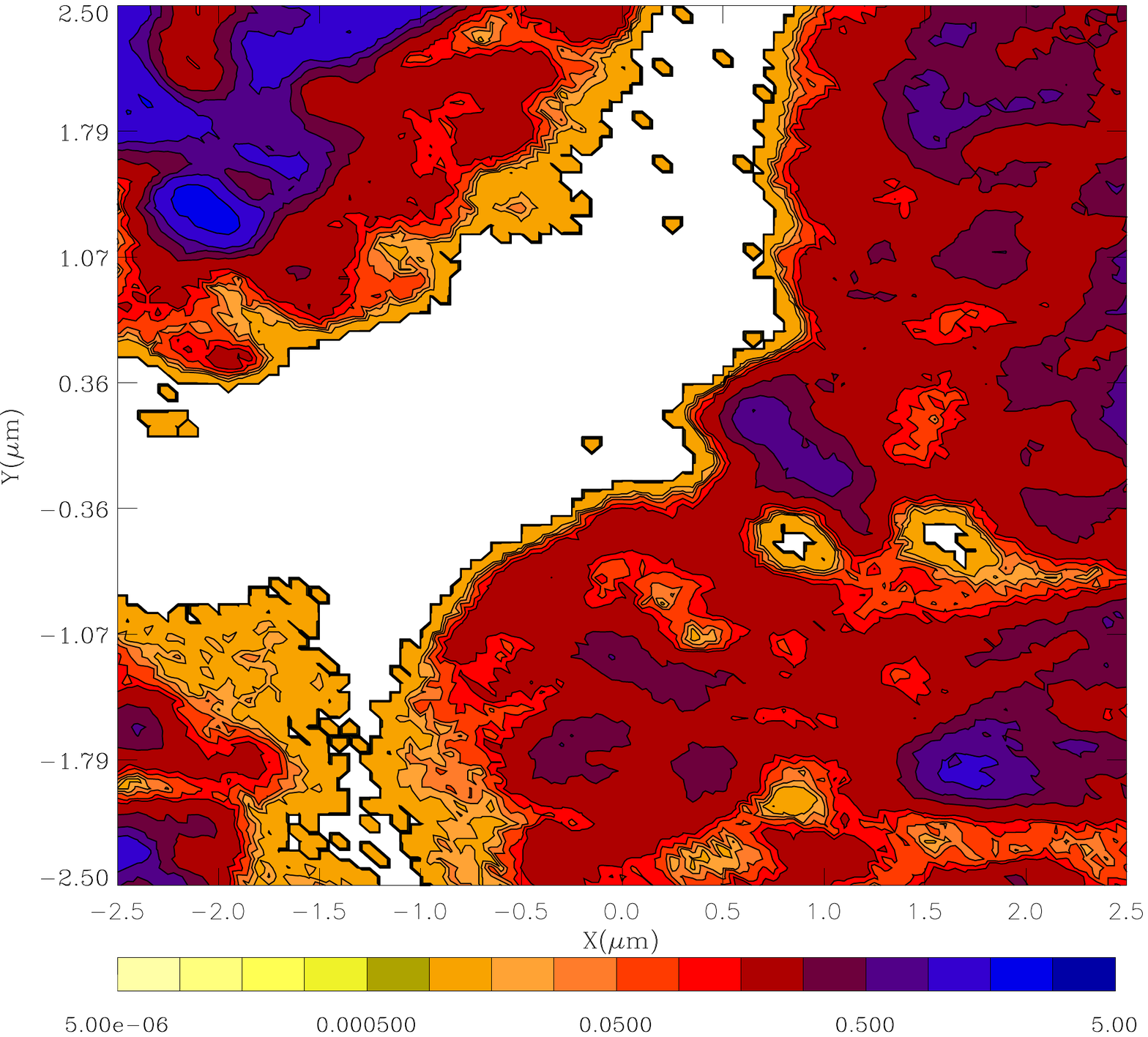}
 }
     \end{center}
\caption{\small{The top panel shows the spatially-averaged projection
    of convolution of learnt material density and correction function,
    onto the image space, learnt using image data $D_1$ in filled
    contours. The real image data $D_1$ is overlaid in black solid
    contours. The lower panel shows the correction functions learnt
    from data $D_1$ and $D_2$ in black and red respectively. The error
    bars in this plot represent the learnt 95$\%$ highest probability
    density region.}}
\label{fig:real_corr}
\end{figure*}

\begin{figure}[!t]
     \begin{center}
  {$\begin{array}{c}
       \includegraphics[width=8cm]{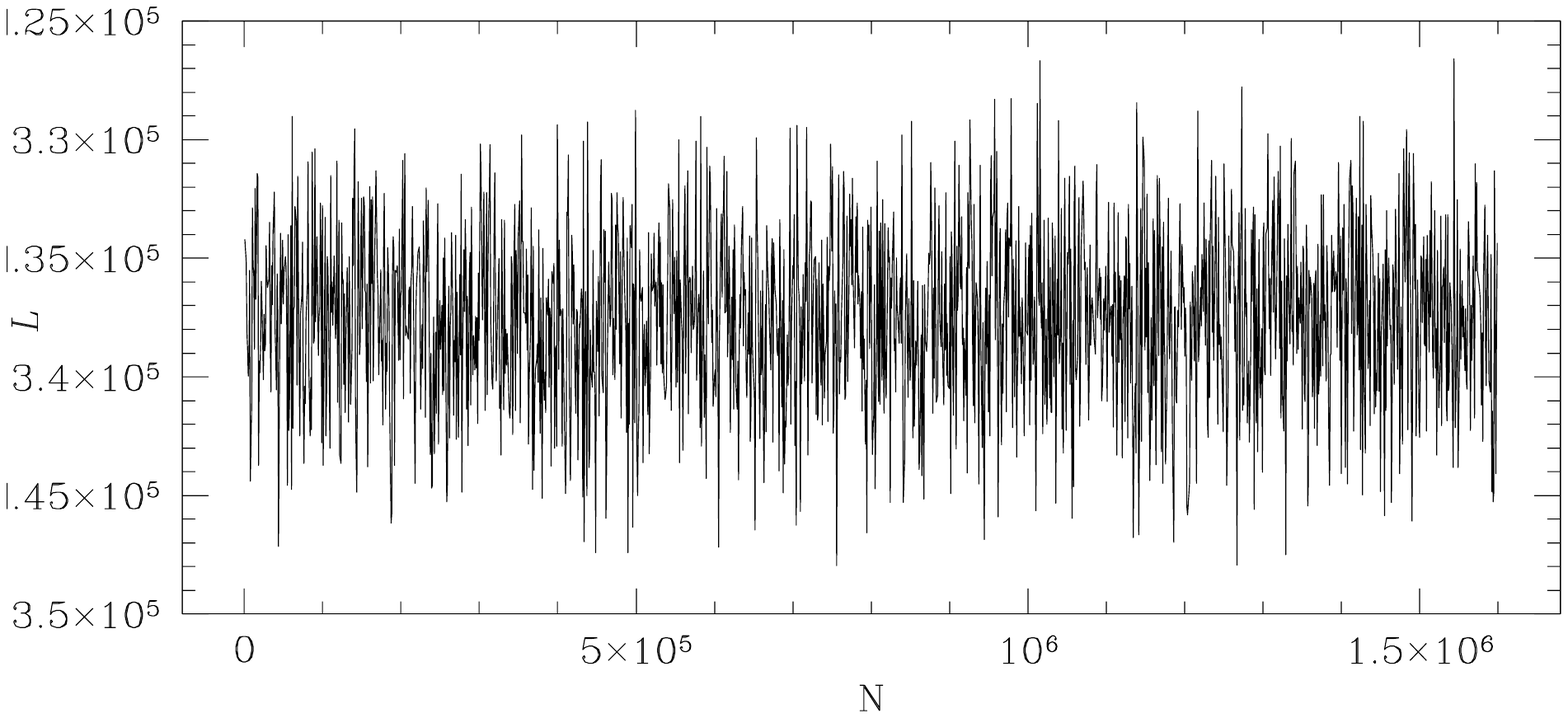} 
  \end{array}$}
     \end{center}
\vspace*{-2.4in}
\caption{\small{The trace of the joint posterior probability density
    of all uknown model parameters given data $D_1$, from a chain that
    is run with these data.}}
\label{fig:inf2}
\end{figure}


\begin{figure}
\vspace*{-1in}
     \begin{center}
  {
       \includegraphics[width=14cm]{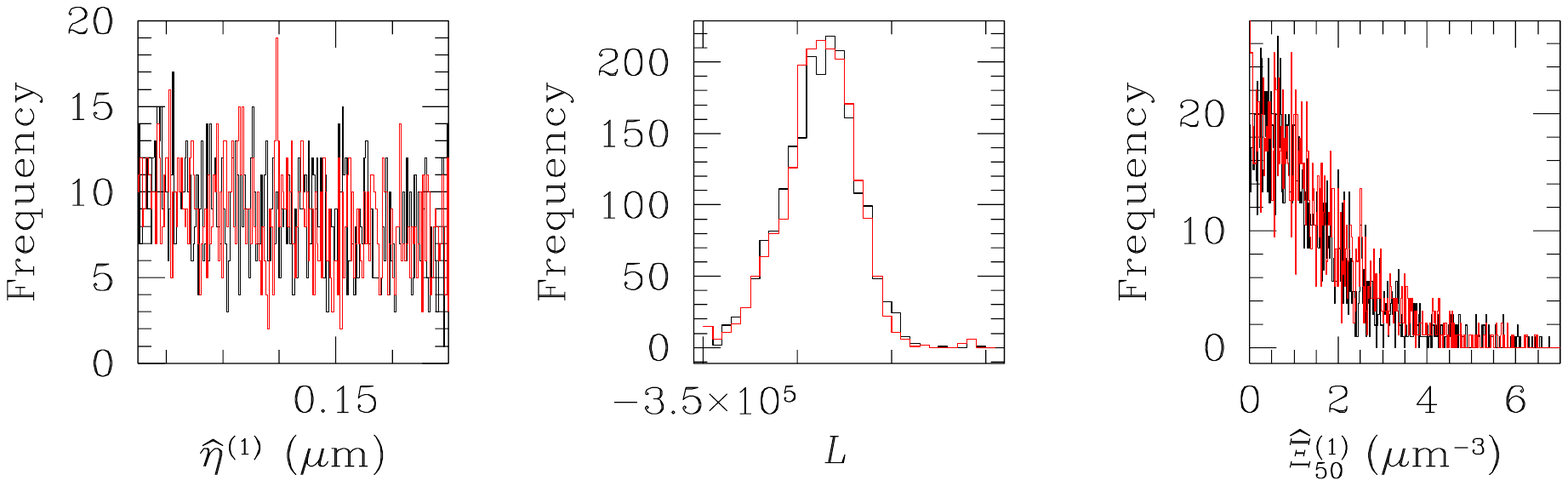}
  }
     \end{center}
\vspace*{-5.2in}
\caption{\small{The histograms of the learnt material density in the voxel
  marked by the viewing vector $\bv_i$, $i=50$ and the lowest value
  of $E$, ($E=\epsilon_1$), from two disjoint parts (middle and end) of the
  MCMC chain that is 1.6$\times$10$^6$ steps long, are depicted in red
  and black respectively, in the right panel. The histograms of the
  learnt correction function for $k=1$ , obtained for these two parts
  of the chain are displayed in the left panel. The histograms of the
  value of the likelihood in these two parts of the chain are plotted
  in red and black in the middle panel. The chain was run with real
  image data $D_1$ .} }
\label{fig:real_corr2}
\end{figure}

\section{Discussion}
\label{sec:discussions}
\noindent
In this work we have advanced a Bayesian methodology that performs
ditribution-free reconstruction of the material density of a material
sample and either a parametric or distribution-free reconstruction of
the microscopy correction function, given 2-D images of the system
taken in any kind of radiation that is generated inside the bulk of
the system as a consequence of the material being impinged upon by an
electron beam. This methodology is advanced as capable of learning the
unknowns even when the the inverse Radon transform is not stable; in
addition to noise in the data, such instability can also arise when
the underlying distribution of the image data is not continuous.  This
in turn arises from a material density function that is heterogeneous
and is marked by a dense or sparse modal structure, characterised by
abruptly declining modal strengths. Given such sharp density contrasts
that typify real-life material samples (see Figure~\ref{fig:real_den}
and Figure~\ref{fig:colour}), mixture models cannot work
satisfactorially in this situation. Also, the density field may not be
necessarily convex in real-life material samples. Crucially, this
inverse problem connects the measured 2-D radiation in any pixel, with
multiple sequential projections of the convolution of the unknown
density and the unknown kernel. The novelty of our solution to this
harder-than-usual inverse problem includes the imaging of the system
at distinct beam energies, adoption of a fully discretised model,
identifying priors on the sparsity of the density and borrowing existing
knowledge from the domain of application. In this application in which
images are taken with electron microscopes, the kernel is the
microscopy correction function, the shapes of which are motivated by
the microscopy literature. In particular, there is deterministic
information available in the literature about the correction at the
surface, allowing for identifiability of the global scales
of the unknown material density and correction function. 

From the point of view of 3-D structure modelling using images taken
with bulk microscopic imaging techniques\footnotemark, (such as
Scanning Electron Microscopy, Electron Probe Microscopy), our aim here
supercedes mere identification of the geometrical distribution of the
material i.e. the microstructure. In fact, we aim to estimate the very
material density at chosen points inside the bulk of the material
sample. Conventionally, Monte Carlo simulation studies of
microstructure are undertaken; convolution of such simulated
microstructure, with a chosen luminosity density function is then
advanced as a model for the density. We advance a methodology that is
a major improvement upon this. One key advantage of our approach is
that estimates of the 3-D material density are derived from
non-invasive and nondestructive bulk imaging techniques. This feature
sets our approach beyond standard methodology that typically relies on
experimental designs involving the etching away of layers of the
sample material at specific depths. Though the microstructure, at this
depth, can in principle be identified this way, a measure of
$\rho(x,y,z)$ is not achievable. That too, only constraints on the
microstructure at such specific depths are possible this way, and
interpolation between the layers - based on assumptions about the
linearity of the microstructure distribution - are questionable in
complex real-life material samples. Of course, such a procedure also
damages the sample in the process. Thus, the scope of the
non-destructive methodology that we advance is superior.  An added
flexibility of our model is that it allows for the learning of the
correction function from the image data.

The expansion of the information is made possible in this work by
suggesting multiple images of the same sample taken with radiation
characterised by different values of the parameter $E$ that controls
the sub-surface depth from which information is carried, to result in
the image. Given the differential penetration depths at different
values of $E$, the images taken in this way are realisations of the
3-D structure of the sample to different depths. While a number of
attempts at density modelling that use multiple viewing angles have
been reported in the literature, the imaging of the system at
different $E$ is less common. Logistical advantages of imaging with
SEM in this way exist over the conventionally undertaken multi-angle
imaging. One example of an imaging technique that may appear to share
similarity with our strategy of collecting information from contiguous
slices at successive values of the depth variable is volume imaging
with MRI, i.e. Magnetic Resonance Imaging
\ctp{hornak,prasad}. However, unlike in our imaging method, in this
technique, the image data at any such slice is obtained via multiple
angle imaging; the imaging parameter whose value is varied to procure
the image data is still the viewing angle, except this variation is
implemented at each individual slice. In our method, the parameter
that is varied is itself the sub-surface depth, achieved by varying
the electron beam energy $E$, since depth has a one-to-one
correspondence with $E$.


 

The method that we have discussed above is indeed developed to solve
an unconventionally difficult deprojection problem, but the method is
equally capable of estimating the unknown density in an integral
equation of the 1$^{st}$ kind - Fredholm or Volterra - and thereby be
applied towards density reconstruction in a wide
variety of contexts, when image synthesis is possible. The scope of
such applications is of course ample, including the identification and
quantification of the density of the metallic molecules that have
infused into a piece of polymer that is employed for charge storage
purposes, degrading the quality of the device as a result, or the
learning of the density of a heterogeneous nanostructure leading to
increased understanding of lack of robustness of device behaviour, or
identification of the distribution of multiple phases of an alloy in
the depth of a metallic sample meant to be used in industry and even
estimation of density of luminous matter in astronomical objects,
viewed with telescopes, at different wavelengths. In fact, in the case
of self-emitting systems that are studied in the emitted particles,
the problem is simpler, (as in several other applications), since the
interaction-volume is not relevant and the integral reduces to the
(easier in general) Volterra integral equation of the first kind. In
either case, domain-specific details need to be invoked to attain
dimensionality reduction.

\section*{Acknowledgments}
We are thankful to Dr. John Aston, Dr. Sourabh Bhattacharya and
Prof. Jim Smith for their helpful suggestions. DC acknowledges the
Warwick Centre for Analytical Sciences fellowship. FR acknowledges
no conflict of interest with any ongoing research work at NVD.

\renewcommand\baselinestretch{1.}
\normalsize

\end{document}